%% file: ANDCoop-J-finalversion.tex
\documentclass[journal]{IEEEtran}

\def\mycmd{2}

\usepackage[english]{babel}
\usepackage{amsmath}		
\usepackage{graphics}		
\usepackage{subcaption}
\usepackage{color,psfrag}
\usepackage{setspace}		
\usepackage{longtable}          
\usepackage{tabularx}
\usepackage{amssymb,amscd,latexsym,dsfont}
\usepackage{float,color,graphicx,balance}
\usepackage{multirow,multicol}
\usepackage{comment,cite}
\usepackage{enumerate}
\usepackage{balance}
\usepackage{stfloats}
\usepackage{tikz}
\usetikzlibrary{fit}					
\usetikzlibrary{backgrounds}	
\usetikzlibrary{matrix}
\usepackage{array,arydshln}
\usepackage{cite}
\usepackage{mathtools}
\usepackage{algpseudocode}
\usepackage[linesnumbered,ruled]{algorithm2e}

\definecolor{refkey}{rgb}{1,0.5,0} 
\definecolor{labelkey}{rgb}{1,0.5,0}

\usepackage{glossaries}
\glossarystyle{listdotted}
\makeglossaries
\renewcommand*{\CustomAcronymFields}{%
  name={\the\glsshorttok},
  description={\the\glslongtok},
  first={\noexpand\emph{\the\glslongtok}\space(\the\glsshorttok)},%
  firstplural={\noexpand\emph{\the\glslongtok\noexpand\acrpluralsuffix}\space(\the\glsshorttok)},%
  text={\the\glsshorttok},%
  plural={\the\glsshorttok\noexpand\acrpluralsuffix}%
}
\usepackage{lipsum} 
\usepackage{mathtools}
\usepackage{cuted}
\usepackage{dsfont}

\if\mycmd 1
\usepackage[font=scriptsize]{caption}
\else
\usepackage[font=footnotesize]{caption}
\fi

\input{include}

\if\mycmd 1

\def \scalevalueS {0.6} 
\def \scalevalueSS {0.46} 
\def \scalevalueSSS {0.5}
\def \scalevalueFigWidth {0.3} 
\def \scalevalueFigWidthSmall {0.74}

\else

\def \scalevalueS {0.9}
\def \scalevalueSS {0.8}
\def \scalevalueSSS {0.6}
\def \scalevalueFigWidth {0.98}
\def \scalevalueFigWidthSmall {0.95}

\fi

\allowdisplaybreaks

\begin{document}
\input{Acronyms}

\bstctlcite{IEEEexample:BSTcontrol}

\if\mycmd 1
\makeatletter
  \def\title@font{\Large}
  \let\ltx@maketitle\@maketitle
  \def\@maketitle{\bgroup%
    \let\ltx@title\@title%
    \def\@title{\resizebox{\textwidth}{!}{%
      \mbox{\title@font\ltx@title}%
    }}%
    \ltx@maketitle%
  \egroup}
\makeatother
\else
\fi

\if\mycmd 1
\title{\Large Exploiting Diversity for  Ultra-Reliable and Low-Latency Wireless Control}
\else
\title{Exploiting Diversity for  Ultra-Reliable and Low-Latency Wireless Control}
\fi%

\if\mycmd 1
\author{{\normalsize Saeed R. Khosravirad, Harish Viswanathan,  and Wei Yu}
\thanks{Saeed R. Khosravirad (saeed.khosravirad@nokia-bell-labs.com), and Harish Viswanathan (harish.viswanathan@nokia-bell-labs.com) are with Nokia Bell Labs, Murray Hill, NJ, USA.}%
\thanks{Wei Yu (weiyu@comm.utoronto.ca) is with the Edward S. Rogers Sr. Department of Electrical and Computer Engineering, University of Toronto, ON, Canada.}}
\else
\author{
    \IEEEauthorblockN{Saeed R. Khosravirad, \textsl{Member, IEEE,} Harish Viswanathan, \textsl{Fellow, IEEE,} and Wei Yu, \textsl{Fellow, IEEE}}
\thanks{Saeed R. Khosravirad (saeed.khosravirad@nokia-bell-labs.com), and Harish Viswanathan (harish.viswanathan@nokia-bell-labs.com) are with Nokia Bell Labs, Murray Hill, NJ, USA.}%
\thanks{Wei Yu (weiyu@comm.utoronto.ca) is with the Edward S. Rogers Sr. Department of Electrical and Computer Engineering, University of Toronto, ON, Canada. The work of Wei Yu is supported by the Nokia Bell Labs and by the Canada Research Chairs program.}
\thanks{This work was presented in part at the IEEE Vehicular Technology Conference (VTC-Spring'19), Kuala Lumpur, Malaysia, May 2019.}}
\fi

\maketitle
\if\mycmd 1
\vspace{-50pt}
\else
\fi

\begin{abstract}

This paper introduces a wireless communication protocol for industrial control systems that uses channel quality awareness to dynamically create network-device cooperation and assist the nodes in momentary poor channel conditions.  To that point,  channel state information  is used to identify nodes with \emph{strong} and \emph{weak} channel conditions. We show that strong nodes in the network are best to be served in a single-hop transmission with transmission rate adapted to their instantaneous channel conditions. Meanwhile,  the remainder of time-frequency resources is used to serve the nodes with weak channel condition using a two-hop transmission with cooperative communication among all the nodes to meet the target reliability in their communication with the controller. We formulate the achievable multi-user and multi-antenna diversity gain in the low-latency regime,  and propose a new scheme for  exploiting those \emph{on-demand}, in favor of reliability and efficiency. The proposed transmission scheme is therefore dubbed adaptive network-device cooperation (ANDCoop), since it is able to adaptively allocate cooperation resources while enjoying  the multi-user diversity gain of the network. 
 We formulate the optimization problem of associating nodes to each group and dividing resources between the two groups.  Numerical solutions show significant improvement in spectral efficiency and system reliability compared to the existing schemes in the literature. System design incorporating the proposed transmission strategy can thus reduce infrastructure cost for future private wireless networks.
\end{abstract}

\if\mycmd 1
\vspace{-10pt}
\else
\fi

\begin{IEEEkeywords}
\if\mycmd 1
\vspace{-10pt}
\else
\fi
Ultra-reliable low-latency communications, factory automation, industrial internet-of-things,  5G
\end{IEEEkeywords}

\if\mycmd 1
\vspace{-10pt}
\else
\fi
\section{Introduction}

\if\mycmd 1
\vspace{-10pt}
\else
\fi
\subsection{Industrial Wireless Control}


Wireless  \gls{iiot} in the next generation of industrial control systems require  communications  with sub-ms, extreme low latency and ``cable-like" ultra-high reliability. For a large-scale network of sensor and actuator devices in factory automation applications, different wireless transmission schemes have recently been proposed to exploit  spatial and multi-user diversity gain in the network. The challenge in this new paradigm of wireless communication is that the design requires  guaranteed service to all, including the \emph{weakest} user, as opposed to the classic paradigm of network design that targets average performance. 
Traditionally, such industrial automation requirements are realized on the factory floor through wired communications e.g., using fieldbus and Ethernet based solutions. The wired solutions, however, are considered to be cumbersome and expensive in many applications. Moreover, the future industrial automation targets a highly  flexible and  dynamic environment of production stations that support robotic mobility to be able to seamlessly re-arrange according to production requirements \cite{Cena:2008,Bennis:2018,Harish:6G2020}. As a result, there is an increased desire to replace wired communication systems for factory automation with wireless alternatives to reduce bulk as well as installation and maintenance costs  \cite{Gungor:2009}. This calls for innovative solutions  in constrast to existing wireless technologies that are designed for  delay tolerant consumer solutions, making them unsuited for  industrial automation \cite{weiner2014design}. In the era beyond the \gls{5g} mobile networks, \gls{urllc} is promised to deliver such demanding requirements using the  advanced  physical layer technologies, including communications in \gls{mmwave} and \gls{uwb} spectrum access \cite{Gilberto:2018,mazgula2020ultra}, accompanied by the improvements in network architecture in  bringing the cloud close to the edge to reduce latency, and using machine learning for a fast and reliable prediction of channels and traffic~\cite{Bennis:2020arxiv}.

\if\mycmd 1
\vspace{-10pt}
\else
\fi
\subsection{Prior Work}

\if\mycmd 1
\vspace{-5pt}
\else
\fi
To fulfill the requirements of ultra-high reliability  within a stringent latency constraint, different diversity techniques are suggested in the literature. Time and frequency diversity techniques \cite{zand2012wireless} as well as spatial diversity and multi-user cooperation \cite{swamy2015cow} can reduce the required reliability-achieving  \gls{snr}. For example, the importance of multi-antenna receive diversity in improving reliability and coverage was pointed out in \cite{brahmi2015gc}. In low-latency industrial automation applications, the  cycle time is shorter than the fading channel coherence time, which rules out viability of \gls{arq}-based time diversity techniques~\cite{Khosravirad:2017}. In \cite{swamy2015cow}, it is shown that relying solely on frequency diversity to achieve $10^{-9}$ error rate requires impractically  high \gls{snr} values in realistic channel conditions. It is further shown that multi-user diversity, even in low or moderate \gls{snr} regime, can  achieve ultra-reliability. Incidentally, increasing transmit diversity by engaging multiple transmitting \glspl{ap}, similar to \gls{comp} technology \cite{randa:2012}, is also not a straightforward path to reliability. In fact, in the absence of \gls{csi} at the transmitter, transmit diversity falls dramatically short of achieving high reliability as discussed in \cite{Rebal:2018}. The study in \cite{Rebal:2018} further shows that by using \gls{csi} to adapt transmission rate at the transmitter, the inherent multi-user diversity gain of a large size network can be exploited to achieve high reliability. To that point, the cooperative transmission in \cite{swamy2015cow} attempts to exploit the full potential of  wireless network by enabling cooperative device-to-device relaying to improve  reliability. The focus of \cite{swamy2015cow} is to devise transmission schemes that do not require transmitter \gls{csi}. Such \gls{csi}-agnostic schemes are not able to take  the  differences in the instantaneous  channel conditions into account, resulting in a sub-optimal and conservative choice of transmission rate determined by the worst user's conditions, and loss of spectral efficiency. On the other hand, in a multi-user wireless network, the overhead for acquiring \gls{csi} can grow large as the size of the network grows. The works in \cite{8259329} and \cite{Rebal:2018} study the impact of such overhead. Particularly, \cite{8259329} considers the overhead of \gls{csi} acquisition when communicating with small packets, thus characterizing the error performance under the finite block length regime. 

System-level  simulations  for multi-user networks under \gls{urllc} requirements is highly time-consuming and complex. We acknowledge that several works in the literature have   addressed and dealt with those complexities, including \cite{Klaus:2018system,7247339}, and have provided insightful conclusions for system design of cellular networks with extreme reliability requirements.


\if\mycmd 1
\vspace{-10pt}
\else
\fi
\subsection{Exploiting Diversity ``On-demand''}
\if\mycmd 1
\vspace{-5pt}
\else
\fi

Spatial diversity transmission, as the most dependable source of achieving high reliability when communicating over fading channels, is a viable solution for reliability that may be achieved by using multiple transmission points or antennas. In low cost deployments, however, it is desirable to have a small number of spatially distributed, simple \glspl{ap} with limited number of antennas. Spatial diversity transmission, however, needs to also exploit multi-user diversity stemming from the fact that several users with different channel conditions are part of the communication system. On the other hand, cooperative relaying among users (as proposed in \cite{swamy2015cow}) essentially also achieves spatial diversity through multi-user diversity. 

The core question this paper  tries to answer is \emph{how to use channel awareness at the transmitter to efficiently allocate radio and cooperation resources, and to exploit diversity according to the instantaneous needs of the users}. We introduce a transmission protocol that is capable of identifying users' channel strength and allows for exploiting different sources of diversity, \emph{on demand}. We propose to adapt the transmission rate for  users with strong channel from the \glspl{ap} according to their channel, while exploiting cooperative diversity for the remaining users with weak channel. This introduces a robust way of deploying the \emph{emergency} resources of network cooperation, only for the devices that absolutely need them. We study the improvements offered by this protocol on the operating spectral efficiency and the minimum required \gls{snr} for reliability. In high \gls{snr}, this impacts  the slope of the outage probability curve and improves diversity order  by deploying network-device cooperation for the poor links. Meanwhile, it brings better multiplexing gain by smartly exploiting the difference in channel conditions.

\if\mycmd 1
\vspace{-15pt}
\else
\fi
\subsection{Channel Estimation in Wireless Networks}
\if\mycmd 1
\vspace{-5pt}
\else
\fi

 
We assume a general transmission framework where data is accompanied with proper amount of pilot signal in all transmissions \cite{hassibi03}. The estimated channel at a receiving node can then be reported back to the transmitting node, e.g. in form of \gls{cqi}. With an adequate frequency of  \gls{cqi} updates, transmitter can improve resource utilization efficiency by adapting the transmission attributes to the channel. In fact, such an approach is widely adopted in multi-user cellular technologies such as the \gls{lte} and \gls{nr}. More interestingly, by assuming channel reciprocity in \gls{tdd} transmission mode for the industrial wireless control problem of our interest with isochronous traffic pattern in  the \gls{dl} and the \gls{ul} directions,  \gls{cqi} acquisition requires no feedback exchange between the nodes. Instead, \gls{cqi} can be  estimated when the node performs channel estimation while  in receiving mode and can be used when the node switches to transmit mode. In this paper, we adopt this assumption and present a new transmission protocol that utilizes channel state information to best exploit sources of diversity in the wireless network.

\if\mycmd 1
\vspace{-15pt}
\else
\fi
\subsection{Contributions}

\if\mycmd 1
\vspace{-5pt}
\else
\fi
The objective of this paper is to design an ultra-reliable transmission scheme for extreme low-latency applications which uses minimal   control signaling for scheduling. To this end, we aim to exploit full spatial and multi-user diversity potential of the network in favor of system reliability and spectral efficiency. Our contributions are summarized as follows.
\if\mycmd 1
\vspace{-5pt}
\else
\fi
\begin{itemize}
\item 	We identify and analyze different sources of diversity gain for ultra-reliable wireless communications in an industrial wireless control network. A network with multiple fully-connected \glspl{ap} is assumed where the \glspl{ap} coordinate their transmissions similar to \gls{comp}. We formulate the achievable multi-user and multi-antenna diversity gain in the low-latency regime,  and propose a new scheme for  exploiting those in favor of reliability and efficiency.
\item 	A new ultra-reliable transmission  scheme dubbed \gls{andcoop} that exploits  different sources of diversity in the network is introduced. The proposed scheme uses the approximate knowledge of \gls{csi} to categorizes the devices into two groups, namely, group of devices with \emph{strong} instantaneous channel, and the group of devices with \emph{weak} channel. The two groups are then scheduled  in separate scheduling phases: first, each of the strong devices  receives its \gls{dl} message with a unique transmission rate that is adapted to its instantaneous channel state; next, the second group of devices are scheduled with a fixed  rate through two-hop cooperative transmission where all the devices in both groups can potentially contribute in as \gls{df} relays.
\item 	Reliability performance  of the proposed transmission scheme is analytically  formulated, in order to characterize  the system outage probability. The analysis is then extended to  diversity-multiplexing trade-off, where closed-form formulations for the  achievable  diversity order are derived. We further formulate the optimization problem of allotting time between the two scheduling phases and provide numerical solutions to the optimization problem.
\item 	Comprehensive and detailed system-level simulations are reported to identify guidelines for optimal system design. The proposed protocol is compared against the existing transmission protocols in the literature. The numerical analysis demonstrates significant concurrent improvement in spectral efficiency (approximately 0.5 \gls{bpcu} per \gls{ap} antenna) and reliability. Alternatively, under fixed spectral efficiency setup, the proposed algorithm acheives the desired reliability at significantly smaller transmit power (around 15 dB improvement compared to the existing schemes), while utilizing around $40\%$ less  relay nodes' energy, which in turn reduces the interference footprint. Moreover, the impact of \gls{csi} estimation error is carefully studied, suggesting that the proposed \gls{andcoop} transmission scheme consistently reduces the impact of such error on system reliability, thanks to the strategy of grouping devices according to channel quality. We identify significant potential in cost reduction for the future private industrial wireless control network, thanks to the improved operation efficiency using the proposed \gls{andcoop} scheme.
\end{itemize}

\if\mycmd 1
\vspace{-15pt}
\else
\fi
\subsection{Organization of the Paper}
\if\mycmd 1
\vspace{-5pt}
\else
\fi

The sequence of this paper is  as follows: in \secref{Sec:ProblemSetup} we present the problem description and the assumed network setup; further, we provide motivations for designing a new ultra-reliable transmission scheme; in \secref{Sec:Model} the proposed channel-aware \gls{urllc} solution is presented and analyzed for outage probability and diversity order; \secref{Sec:Results} presents and discusses the numerical analysis of the proposed scheme; and finally, \secref{Sec:Conclusion} covers the concluding remarks.

\section{Problem Setup}
\label{Sec:ProblemSetup}

\begin{figure}[t]
\begin{center}
\psfrag{u}[c][c][\scalevalueSS]{wireless link}
\psfrag{v}[c][c][\scalevalueSS]{wired link}
\psfrag{x}[c][c][\scalevalueSS]{wireless connected}
\psfrag{w}[c][c][\scalevalueSS]{devices}
\psfrag{y}[c][c][\scalevalueSS]{access points}
\psfrag{z}[c][c][\scalevalueSS]{controller}
\includegraphics[width=\scalevalueFigWidth\columnwidth,keepaspectratio]{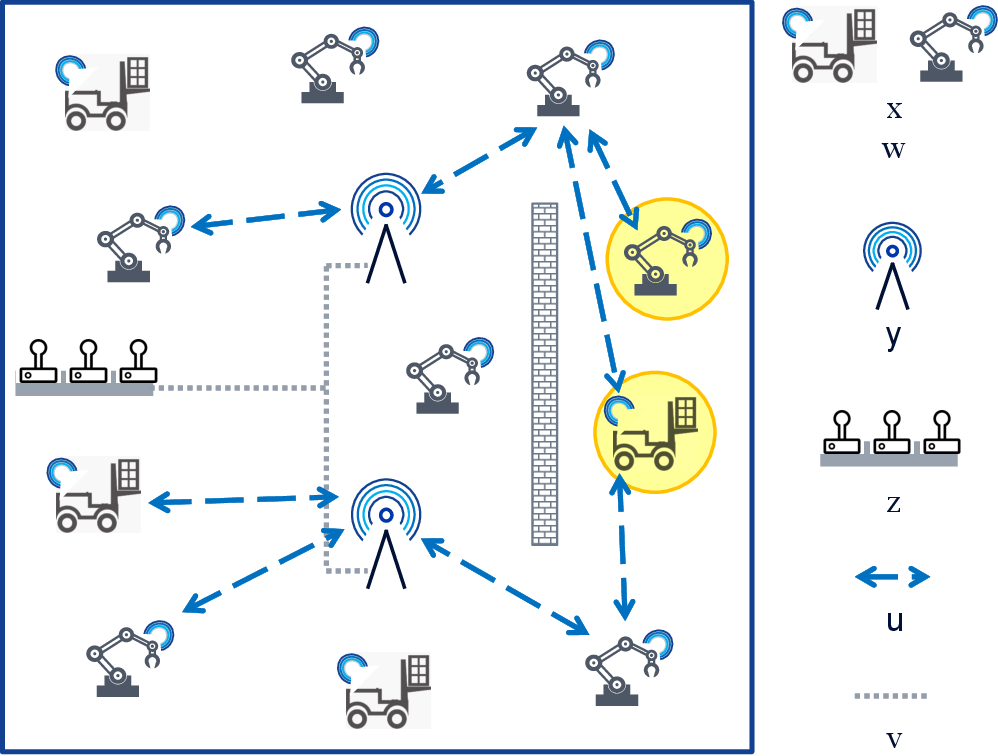}
\caption{Wireless network model for industrial wireless control.  Wireless devices with \emph{weak} channel conditions (highlighted) are identified for  two-hop communication.}
\label{Fig:Model}
\end{center}
\if\mycmd 1
\vspace{-35pt}
\else
\fi
\end{figure}

In this section, we first describe the communications system model of interest and  highlight the main system assumptions we use in our analysis. Then, we discuss exploiting diversity gain a multi-user wireless   network under the  paradigm of ultra-reliable communications to motivate our design target in exploiting full diversity potential of the network for industrial wireless control.

\if\mycmd 1
\vspace{-10pt}
\else
\fi

\subsection{System Model}

\paragraph*{Network} $N$ devices are scattered on a factory floor and are wirelessly connected with the controller \glspl{ap}. \figref{Fig:Model} illustrates the considered network where a controller is wired to $M$ fully synchronized \glspl{ap}. This paper considers  a \gls{comp} setting  where all \glspl{ap} are synchronized and they coordinate their transmission attributes for transmission to every device.  All the communicating nodes have a single antenna for transmission and reception. Every device expects an independent   $B$ bytes of  data to be delivered every $T$ seconds over a bandwidth of $W$ Hertz. We use $\eta = \frac{N B}{T W}$, measured in \gls{bpcu}, to denote the overall spectral efficiency of the system. 
Let $\msa = \set{1,2,\dots,M}$ denote  the set of \glspl{ap}, where $M = \card{\msa}$ is the number of \glspl{ap}. Similarly, let $\msd = \set{1,2,\dots,N}$ be the set of device IDs where $N = \card{\msd}$. Throughout the paper, we reserve the letter $R$ to denote transmit rate measured in \gls{bps}.

\paragraph*{Channel dynamics} Wireless channels linking every \gls{ap}-device and device-device pair are assumed to undergo independent frequency-flat Rayleigh fading. We note that this assumption is adopted for analytical tractability, although measurement campaigns for industrial environments show frequency-selectivity over wide bandwidth \cite{nist1951,rapp91}. We assume a  setting where each time-cycle experiences a constant channel which fades independently from one cycle to the next. Let $\ha{i}{j}$ and $\ra{i}{j}$ denote the channel fade random variable and  the average received \gls{snr} (which  includes the effect of path loss and is averaged with respect to  fading distribution) of the transmission from  \gls{ap} $i$ to device $j$, where $i \in \msa$ and $j \in \msd$.  We use $\ga{i}{j} = \ra{i}{j}|\ha{i}{j}|^2$ to denote the instantaneous received \gls{snr}. Similarly, let $\hd{k}{j}$, $\rd{k}{j}$ and $\gd{k}{j}$ denote the same variables for the link  from device $k$ to device $j$, where $k, j \in \msd$. Note that  $\ra{i}{j} = \Pa / (W \cdot \sigma_0)$ and $\rd{k}{j} = \Pd / (W \cdot \sigma_0)$, where $\Pa$ and $\Pd$ denote the transmit power of an access point and a device, respectively, and $\varnoise$ denotes \gls{psd} of the \gls{awgn}.

%
%


\paragraph*{Outage model}

A device is said to be in outage if the transmission rate $R$ exceeds the instantaneous channel capacity, and  is considered successful otherwise. We assume  distributed space-time coding that collects spatial diversity through summation of the received signal powers from all transmitters. Therefore, with $\msc_j$ denoting the set of nodes cooperatively transmitting with rate $R$ to node $j$ over bandwidth $W$, the transmission  fails if
\begin{align}\label{Eq:linkoutage}
W \log \left( 1 + \sum_{i \in \msa \cap \msc_j}\ga{i}{j} + \sum_{k \in \msd \cap \msc_j}\gd{k}{j} \right) < R.
\end{align}
The expression in \eqref{Eq:linkoutage} implicitly assumes long block-length transmission. While admittedly, the transmission  of small packets in line with what is typically expected in \gls{urllc} scenarios and also suitable for the block-fading model challenges the assumption that the packets are long enough for \eqref{Eq:linkoutage}, we note that the impact of such assumptions can be further evaluated by adopting the  finite block-length regime outage models \cite{PolyanskiyFBL:2010}. More importantly, the recent findings in \cite{6802432} suggest  that in a fading channel, the effect of outage dominates the effect of short block-length, so the  outage capacity is in fact a fair substitute  for the finite block-length   fundamental limits. For this reason, the rest of this paper focuses on the outage model in \eqref{Eq:linkoutage}. 

Similar to the previous works in \cite{Rebal:2018TCOM,swamy2015cow,Khosravirad:vtc2019}, we analyze  \emph{system outage probability} as the key performance metric, denoted by $\Pout(.)$ and defined as the  probability that at least one device fails to decode its own message at the end of time-cycle $T$. This is a more appropriate measure for reliability of communication in an industrial wireless control setup compared to e.g., average outage probability across devices. The argument is that the industrial wireless control system may only continue its operation when all devices follow the controller instructions, and the system fails if at least one devices fails. Note that $\Pout(.)$ is a function of the channel random variables, as well as the  parameters of the system. Moreover, such definition complies with  the joint definition of reliability and latency requirements in the context of \gls{urllc}. In essence, a \gls{urllc} system satisfies its requirements only if it can guarantee the desired reliability level \emph{within} the desired latency budget \cite{3gpp38913}. Therefore, in this work, instead of the statistics of the experience  delay, we are interested in the  outage probability within a constrained latency of $T$ seconds.

\paragraph*{Diversity-multiplexing}

It is widely accepted that the  end goal of \gls{urllc} systems is to increase reliability, and therefore, the system outage probability curve is the natural benchmark for performance evaluation. However, the true performance  of such system can only be evaluated if data rate is monitored alongside the reliability. Thanks to the  choice of \emph{system outage probability} (described above) to represent error rate in the system model, the diversity gain can be captured as the slope at which the error rate decays in the high \gls{snr} regime. Moreover, we define the  multiplexing gain $r$ as the ratio at which the payload size per device $B$ increases with transmit power $\Pt$ in log scale, i.e., $B \propto r \log \Pt$. Thus, the dual benefits can be captured by the diversity-multiplexing tradeoff in the high \gls{snr} regime, where similar to \cite{Zhao:2007diversity,Tse2003divmux}, we say that a diversity gain $d(r)$ is achieved at multiplexing gain $r$, if  $\eta = r \log \Pt$ and,
\begin{align}\label{Eq:dive_order_def}
d(r) = - \lim_{\Pt \rightarrow \infty} \frac{\log \Pout(r \log \Pt)}{\log \Pt},
\end{align}
thus capturing the tradeoff between data rate increase (i.e., $r$) and diversity order, in high \gls{snr}.

\input{table.tex}

\paragraph*{Channel estimation} We assume instantaneous \gls{csi} of the \gls{ap}-device pairs are present at the controller in the form of link strength $\ga{i}{j}$'s, which doesn't require the knowledge of channel phase. This can be for instance provided by frequent transmission of uplink pilot sequences by the devices similar to \gls{srs} in \gls{lte}. 
Each \gls{ap}  estimates its channel from all the devices, using those pilot sequences. The \gls{csi} will be used to identify groups of devices with \emph{strong} and \emph{weak} channel conditions and to adapt the transmission rate for the former group. The variance of channel estimation error can be arbitrarily minimized by increasing the number of pilot sequences and transmit power of the pilots \cite{hassibi03,yoo06,Rebal:2018}. 

Assuming that the channel \gls{snr}, $\rho$, is known, we use $\hat{h}$ and $\hat{g}$  to denote the estimated  channel fade and estimated \gls{snr}, respectively. We further use  $L$ to denote the length of the pilot training sequence for each device, with duration of $T_P = L \cdot T_S$ seconds over orthogonal time-frequency resources, where $T_S = 1/W$ is the symbol period. The total overhead cost of pilot transmission then equals to $N \cdot T_P$, and is paid out of the time budget $T$, leaving $T_D = T - N\cdot T_P$ seconds for data transmission. Using recursive \gls{mmse} channel estimation \cite{hassibi03,yoo06}, the true Rayleigh fade $h$ can be written as $h=\hat{h}+\esterror$, where $\esterror\sim\mathcal{CN}(0,\sigma_e(L))$, $\hat{h}\sim\mathcal{CN}(0,1-\sigma_e(L))$, and
\begin{align}\label{eq_mmse}
\sigma_e(L)=\frac{1}{1+L\cdot\rho}.
\end{align}

With respect to channel estimation, and for completeness of the investigation, we adopt two scenarios in this paper; namely,  genie-aided \gls{pcsi}, i.e., where fade is perfectly estimated as $\hat{h}=h$, at the cost of zero pilot overhead $L = 0$, leaving $T_D = T$; and the case of \gls{icsi}, where channel estimation error is a function of the pilot training sequence length $L$ based on \eqref{eq_mmse}, resulting in $T_D = T - N \cdot L \cdot T_S$.


\paragraph*{Notations} Throughout the paper, we use the notations listed in \tabref{table:notation}.

\if\mycmd 1
\vspace{-10pt}
\else
\fi
\subsection{On the Role of Multi-User Diversity in Low-Latency Regime}

In a \emph{large} network with multiple users, each fading independently, there is likely to be a user whose channel is near its peak, at any time.  This can be utilized to maximize the long term total throughput  by use of  \gls{csi} feedback and always serving the user with the strongest channel \cite{Knopp:1995,Grossglauser:2002} hence, exploiting \emph{multi-user diversity} gain. 
 Similarly, for a given spectral efficiency, the per-user reliability of  transmission can be maximized by choosing the user with strongest channel at any time. Therefore, with loose latency requirement, multi-user diversity gain is a natural source of reliability and efficiency. However, in low-latency regime, where tolerated latency is smaller or equal to the channel coherence time, it is likely to have one or few  users whose channels are poor, at any time. It is therefore challenging  to exploit multi-user diversity while guaranteeing timely reliability to multiple users with asymmetric channel statistics.

To further analyze the  diversity gain in low-latency regime, let's assume the network setup described earlier in this section with $M = 1$, where all the channel gain $\ga{i}{j}$'s are perfectly known and thus, the controller can precisely determine the achievable rate $\aR_{j}$ for device $j$, and the  \gls{ap} targets an average spectral efficiency  of $\frac{N B}{T W}$ in each time cycle $T$, with equal packet size for every scheduled device. Let's consider the case where the scheduler has the complete freedom to choose any nonempty subset of the $N$ devices in each time cycle.  We model the average latency based on $K$, the number of users that are scheduled within a given time cycle. The average experienced latency by a device to be scheduled can be shown to be $\frac{(N+K)}{2K}T$\footnote{Derived by averaging across all devices, knowing that the first $K$ devices experience latency of $T$, while the last $K$ devices to be scheduled in a  round experience latency of $NT/K$.}. For example under round-robin scheduling, by scheduling all devices in every time cycle $T$, i.e., $K = N$, the average latency is $T$, and by scheduling only one device in each time cycle the average latency is increased to $\frac{(N+1)}{2}T$. The scheduler is thus rewarded with reduced average latency, for scheduling every additional device out of the $N$. The cost of scheduling an additional device is going to be a loss in the collected multi-user diversity order. The following proposition addresses the trade-off between average latency and the collected multi-user diversity gain, assuming \gls{iid} Rayleigh fading on every link. 
\if\mycmd 1
\vspace{-5pt}
\else
\fi
\begin{proposition}
\label{Prop:motiv1}
The maximum diversity order exploited at zero multiplexing point by the scheduler described above for scheduling exactly $K$ users in each time cycle is $N - K + 1$ for the case of $M=1$, and $M(N-K+1)$ for general $M$.
\end{proposition}
\if\mycmd 1
\vspace{-10pt}
\else
\fi
\begin{proof}
See \appref{App:motiv}.
\end{proof}
In other words, with strict latency requirements, i.e., all the $N$ devices must be scheduled over one coherence time, which is the case for an industrial wireless control network as described earlier in this section, then the system experiences diversity order of $1$ (for the case of M=1), meaning that no multi-user diversity gain is exploited. On the other hand, the maximum multi-user diversity order of $N$,  is only exploited when the latency requirement is maximally loose and the controller gets to schedule the best device in each time cycle. Midway, the controller can trade off latency with  diversity order by  transmitting to $K \leq N$ devices over each coherence time, thus gaining the diversity order of $N - K + 1$. In fact, as we  see in the \secref{Sec:Model}, the transmission protocol proposed in this paper  benefits from such  trade-off in exploiting the multi-user diversity gain by scheduling a subset of the devices with the highest channel strengths. For the remaining devices, the protocol  seeks for the gain of cooperation among nodes, which is the topic of the following discussion.

\if\mycmd 1
\begin{figure*}[t]
\begin{center}
\begin{subfigure}{.25\textwidth}
\begin{center}
\begin{psfrags}
\psfrag{xlabel}[c][c][\scalevalueSSS]{X (m)}
\psfrag{ylabel}[c][c][\scalevalueSSS]{Y (m)}
\psfrag{title}[c][c][\scalevalueSSS]{\gls{snr} (dB)}
\psfrag{xlabel2}[c][c][\scalevalueSSS]{X (m)}
\psfrag{ylabel2}[c][c][\scalevalueSSS]{Y (m)}
\psfrag{title2}[c][c][\scalevalueSSS]{coverage map}
\includegraphics[width=0.6\columnwidth,keepaspectratio]{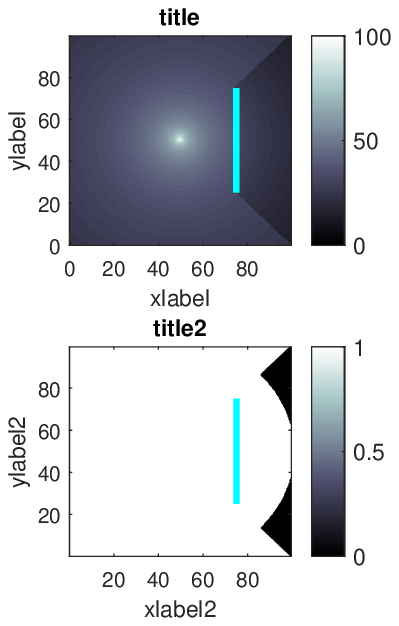}
\end{psfrags}
\caption{Coverage of single-hop transmission.}
\label{fig:coverage_1h}
\end{center}
\end{subfigure}%
\begin{subfigure}{.55\textwidth}
\begin{subfigure}{.45\textwidth}
\begin{center}
\begin{psfrags}
\psfrag{xlabel}[c][c][\scalevalueSSS]{X (m)}
\psfrag{ylabel}[c][c][\scalevalueSSS]{Y (m)}
\psfrag{title}[c][c][\scalevalueSSS]{\gls{snr} (dB)}
\psfrag{xlabel2}[c][c][\scalevalueSSS]{X (m)}
\psfrag{ylabel2}[c][c][\scalevalueSSS]{Y (m)}
\psfrag{title2}[c][c][\scalevalueSSS]{coverage map}
\psfrag{xlabel3}[c][c][\scalevalueSSS]{X (m)}
\psfrag{ylabel3}[c][c][\scalevalueSSS]{Y (m)}
\psfrag{title3}[c][c][\scalevalueSSS]{\gls{snr} (dB)}
\psfrag{xlabel4}[c][c][\scalevalueSSS]{X (m)}
\psfrag{ylabel4}[c][c][\scalevalueSSS]{Y (m)}
\psfrag{title4}[c][c][\scalevalueSSS]{coverage map}
\includegraphics[width=.6\columnwidth,keepaspectratio]{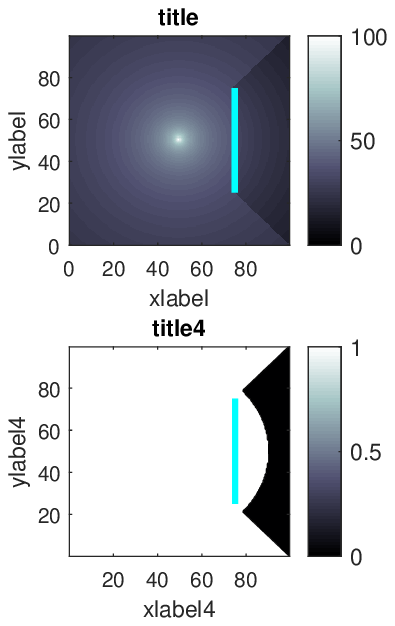}
\end{psfrags}
\end{center}
\end{subfigure}%
\begin{subfigure}{.45\textwidth}
\begin{center}
\begin{psfrags}
\psfrag{xlabel1}[c][c][\scalevalueSSS]{X (m)}
\psfrag{ylabel1}[c][c][\scalevalueSSS]{Y (m)}
\psfrag{title1}[c][c][\scalevalueSSS]{\gls{snr} (dB)}
\psfrag{xlabel2}[c][c][\scalevalueSSS]{X (m)}
\psfrag{ylabel2}[c][c][\scalevalueSSS]{Y (m)}
\psfrag{title2}[c][c][\scalevalueSSS]{\gls{snr} (dB)}
\psfrag{xlabel3}[c][c][\scalevalueSSS]{X (m)}
\psfrag{ylabel3}[c][c][\scalevalueSSS]{Y (m)}
\psfrag{title3}[c][c][\scalevalueSSS]{\gls{snr} (dB)}
\psfrag{xlabel4}[c][c][\scalevalueSSS]{X (m)}
\psfrag{ylabel4}[c][c][\scalevalueSSS]{Y (m)}
\psfrag{title4}[c][c][\scalevalueSSS]{coverage map}
\includegraphics[width=.6\columnwidth,keepaspectratio]{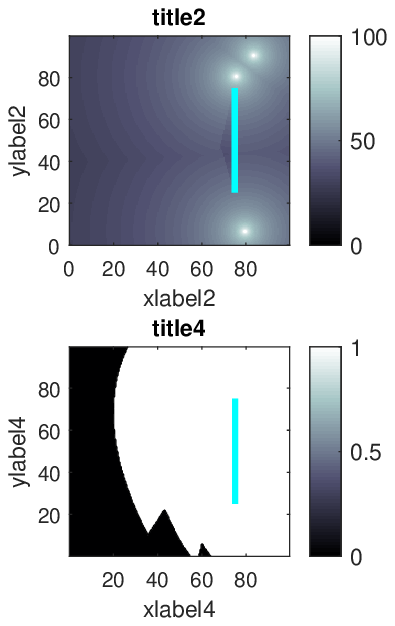}
\end{psfrags}
\end{center}
\end{subfigure}%
\caption{Coverage of two-hop transmission with cooperative relaying.}
\label{fig:coverage_2h}
\end{subfigure}%
\end{center}
\vspace{-15pt}
\caption{Map of \gls{snr} and coverage of $10^{-9}$ outage probability for spectral efficiency of 1 \gls{bpcu}, in an area of $100\times100$ $m^2$, in presence of a blockage (cyan color wall); the bright point in the center locates the \gls{ap} while the rest of the bright points locate the relay devices; (a)  \gls{ap} with 4 antennas in the center provides $95.5\%$ when transmitting over the full time cycle $T$, i.e., transmission rate of 1 \gls{bpcu}; (b) on the left,  $87.95\%$ coverage from an \gls{ap} with 4 antennas  at transmission rate of 2 \gls{bpcu}, over first $T/2$ duration; on the right, $72.1\%$ coverage using three single antenna relay devices at transmission rate of 2 \gls{bpcu}, over the second $T/2$ duration; $100\%$ combined coverage of the two phases.}
\label{fig:coverage_map}
\if\mycmd 1
\vspace{-20pt}
\else
\vspace{-10pt}
\fi
\end{figure*}
\else
\begin{figure*}[t]
\begin{center}
\begin{subfigure}{.35\textwidth}
\begin{center}
\begin{psfrags}
\psfrag{xlabel}[c][c][\scalevalueSSS]{X (m)}
\psfrag{ylabel}[c][c][\scalevalueSSS]{Y (m)}
\psfrag{title}[c][c][\scalevalueSSS]{\gls{snr} (dB)}
\psfrag{xlabel2}[c][c][\scalevalueSSS]{X (m)}
\psfrag{ylabel2}[c][c][\scalevalueSSS]{Y (m)}
\psfrag{title2}[c][c][\scalevalueSSS]{coverage map}
\includegraphics[width=0.6\columnwidth,keepaspectratio]{snr_map_1hop.eps}
\end{psfrags}
\caption{Coverage of single-hop transmission.}
\label{fig:coverage_1h}
\end{center}
\end{subfigure}%
\begin{subfigure}{.6\textwidth}
\begin{subfigure}{.5\textwidth}
\begin{center}
\begin{psfrags}
\psfrag{xlabel}[c][c][\scalevalueSSS]{X (m)}
\psfrag{ylabel}[c][c][\scalevalueSSS]{Y (m)}
\psfrag{title}[c][c][\scalevalueSSS]{\gls{snr} (dB)}
\psfrag{xlabel2}[c][c][\scalevalueSSS]{X (m)}
\psfrag{ylabel2}[c][c][\scalevalueSSS]{Y (m)}
\psfrag{title2}[c][c][\scalevalueSSS]{coverage map}
\psfrag{xlabel3}[c][c][\scalevalueSSS]{X (m)}
\psfrag{ylabel3}[c][c][\scalevalueSSS]{Y (m)}
\psfrag{title3}[c][c][\scalevalueSSS]{\gls{snr} (dB)}
\psfrag{xlabel4}[c][c][\scalevalueSSS]{X (m)}
\psfrag{ylabel4}[c][c][\scalevalueSSS]{Y (m)}
\psfrag{title4}[c][c][\scalevalueSSS]{coverage map}
\includegraphics[width=.7\columnwidth,keepaspectratio]{snr_map_2hop_2.eps}
\end{psfrags}
\end{center}
\end{subfigure}%
\begin{subfigure}{.5\textwidth}
\begin{center}
\begin{psfrags}
\psfrag{xlabel1}[c][c][\scalevalueSSS]{X (m)}
\psfrag{ylabel1}[c][c][\scalevalueSSS]{Y (m)}
\psfrag{title1}[c][c][\scalevalueSSS]{\gls{snr} (dB)}
\psfrag{xlabel2}[c][c][\scalevalueSSS]{X (m)}
\psfrag{ylabel2}[c][c][\scalevalueSSS]{Y (m)}
\psfrag{title2}[c][c][\scalevalueSSS]{\gls{snr} (dB)}
\psfrag{xlabel3}[c][c][\scalevalueSSS]{X (m)}
\psfrag{ylabel3}[c][c][\scalevalueSSS]{Y (m)}
\psfrag{title3}[c][c][\scalevalueSSS]{\gls{snr} (dB)}
\psfrag{xlabel4}[c][c][\scalevalueSSS]{X (m)}
\psfrag{ylabel4}[c][c][\scalevalueSSS]{Y (m)}
\psfrag{title4}[c][c][\scalevalueSSS]{coverage map}
\includegraphics[width=.7\columnwidth,keepaspectratio]{snr_map_2hop_1.eps}
\end{psfrags}
\end{center}
\end{subfigure}%
\caption{Coverage of two-hop transmission with cooperative relaying.}
\label{fig:coverage_2h}
\end{subfigure}%
\end{center}
\caption{Map of \gls{snr} and coverage of $10^{-9}$ outage probability for spectral efficiency of 1 \gls{bpcu}, in an area of $100\times100$ $m^2$, in presence of a blockage (cyan color wall); the bright point in the center locates the \gls{ap} while the rest of the bright points locate the relay devices; (a)  \gls{ap} with 4 antennas in the center provides $95.5\%$ when transmitting over the full time cycle $T$, i.e., transmission rate of 1 \gls{bpcu}; (b) on the left,  $87.95\%$ coverage from an \gls{ap} with 4 antennas  at transmission rate of 2 \gls{bpcu}, over first $T/2$ duration; on the right, $72.1\%$ coverage using three single antenna relay devices at transmission rate of 2 \gls{bpcu}, over the second $T/2$ duration; $100\%$ combined coverage of the two phases.}
\label{fig:coverage_map}
\if\mycmd 1
\vspace{-20pt}
\else
\vspace{-10pt}
\fi
\end{figure*}
\fi

\if\mycmd 1
\vspace{-10pt}
\else
\fi
\subsection{Motivation for Multi-Hop Transmission}

Cooperative relaying has been studied recently in several works as an enabler of \gls{urllc}, e.g., see \cite{swamy2015cow,swamy:2018relsel,Arvin:2019,Kaiming:2019ratespliting}, leveraging on the spatial diversity gain from engaging multiple relay devices, which increases robustness to fading variations. The focus of the design in cooperative relaying scheme in \cite{swamy2015cow} has been to mitigate the effect of small-scale fading. Such approach is highly beneficial in  absence of \gls{mimo} techniques. However, with increasing deployment of  massive \gls{mimo}, in practice,  those benefits  can be largely undermined. Particularly, cellular communication technologies typically rely on channel hardening effect of \gls{mimo} to mitigate the effect of small-scale fading \cite{Hochwald2004channelhardening}. 

Nevertheless, multi-hop relaying has historically been considered as a means of extending coverage in various wireless technologies, such as \gls{wimax}, \gls{lte}, and  recenlty in \gls{5g} \gls{nr}, e.g., see \cite{kim2008relaycoverage,Gui2018relaycoverage,Lang2009relaycoverage,DawidKoziol:2017,3gpp22866}. The coverage problem caused by static or mobile blockages is in fact a challenging issue in industrial environments. Blockage can generally impede the link \gls{snr} by obstructing the \gls{los}. More significantly, when blockage size is several times larger than the electromagnetic  wavelength, the diffraction around the obstacle becomes weaker, making the impact of blockage stronger. Consequently, blockages becomes a more severe challenge in higher frequencies, or the so called \gls{mmwave} \cite{5262296}. Moreover, the dynamic nature of factory floor with large number of static and moving objects makes it difficult to provide every time/everywhere wireless link availability. This further motivates the use of cooperative relaying to deal with temporary and/or zonal loss of coverage.

To this point, the example in \figref{fig:coverage_map} illustrates the \gls{snr} and coverage of $10^{-9}$ outage probability for spectral efficiency of 1 \gls{bpcu} over an area of $100\times100$ $m^2$. A simple static blockage is positioned on the right side of the area. First, in \figref{fig:coverage_1h}, it is shown that using the total time budget $T$ (i.e., transmitting at  1 \gls{bpcu}), direct transmission by the \gls{ap} with 4 antennas provides $95.5\%$ coverage. This means that only around $95.5\%$ of the points over the area can achieve the  required target outage probability of $10^{-9}$ with a single-hop transmission. Then, in \figref{fig:coverage_2h}, the time budget $T$ is divided by two, and the transmission is done twice at the doubled rate of 2 \gls{bpcu}. Therefore, the coverage of the direct transmission from access point reduces to $87.95\%$ (left hand side of \figref{fig:coverage_2h}), due to the increase in transmission rate. However, in the second $T/2$ portion of the time, three  devices randomly positioned around the blockage, and each  having a single transmit antenna,  relay the transmission from the \gls{ap}. The relaying phase at  2 \gls{bpcu}, also provides a complementary coverage of $72.1\%$, mostly around the area affected by the blockage. But more interestingly, the overall coverage of the two phases reaches the desirable $100\%$.

Knowing that the blockage affects the coverage around the right hand side  of the square area, such coverage enhancement from two-hop relaying can in practice be directed  towards devices in the same area. This in fact increases the efficiency of ultra-reliable communications, by deploying the cooperative relaying in an \emph{on-demand} fashion, only for the devices with coverage issues. 

Improving coverage for \gls{iiot} applications can alternatively be done by densification of the \glspl{ap}. However, over-provisioning is not  efficient in terms of the cost of the network deployment. It should be noted that the intention of the example in \figref{fig:coverage_map} is not to claim that two-hop cooperative  relaying is always better than single-hop transmission. In fact, as we will discuss in the following sections, with a smart and dynamic algorithm  to use the cooperative relaying gain in an \emph{on-demand} manner, the overall spectral efficiency can be improved compared to the case where all the devices are served with two-hop transmission.

\if\mycmd 1
\vspace{-10pt}
\else
\fi
\section{Proposed Channel-Aware Transmission Protocol}
\label{Sec:Model}

In this section we first introduce the proposed transmission protocol. Then outage and diversity order analysis of the protocol are presented.

\if\mycmd 1
\vspace{-10pt}
\else
\fi
\subsection{Transmission Protocol}
\label{Sec:transmissionprotocol}

The total \gls{dl} transmission time $T_D$ is divided into two parts, using a partitioning factor $0 \leq \beta \leq 1$, where $\Toh = \beta T_D$ is used for rate-adaptive single-hop transmission of independent messages to devices with strong instantaneous channel to the controller \glspl{ap} in a \gls{tdma} fashion. The remaining $\Tth = (1-\beta)T_D$ is used for  two-hop cooperative transmission phase to the rest of the devices, where the messages of all remaining devices are aggregated and transmitted in two hops. In addition, the controller acquires the knowledge of $\gha{i}{j}$'s for the \gls{ap}-device pairs via channel estimation using pilots of length $L$. Based on this, devices are put in an order according to the instantaneous transmission rate that they can receive at from the \glspl{ap}. The achievable transmission rate for device $j$, $\aR_{j}$, can be computed  using \eqref{Eq:linkoutage}. The controller estimates the achievable rate, using
\begin{align}\label{Eq:achievablerate}
\ahR_{j} = W \log \left( 1 + \sum_{i \in \msa}\gha{i}{j}  \right).
\end{align}
Note that, for \gls{icsi}, i.e., $\sigma_{e} > 0$,  $\prob{\ahR_{j} > \aR_{j}} = 0.5$. In other words, regardless of the channel  estimation  precision, the estimated transmission rate is above the achievable rate $50\%$ of the time. To curb the impact of channel estimation error in case of \gls{icsi}, we use a rate back-off parameter $0 \leq \rbo \leq 1$ to adjust the transmission rate to $\rbo \cdot \ahR$. The following definition denotes the largest  subset of the devices  that the controller can accommodate with single-hop transmission  over   time $\tau$. 
\begin{align}\label{Eq:strongdevices}
\pmb{S}(\tau, \{\gha{i}{j} \}) = \;  \underset{\mss}{\arg} \; \max_{\mss \subset \msd} \; \; & \card{\mss} \\  \nonumber
 \text{subject to} & \sum_{j \in \mss} \frac{1}{\rbo \cdot \ahR_{j}} \leq \frac{\tau}{B}, \\
 &   \ahR_{k} \leq \ahR_{j}, \; \forall k \in \msd \setminus \mss, \forall j \in \mss \nonumber
\end{align}
Let $\msoh$, where $\msoh \subset \msd$, and $\msth = \msd \setminus \msoh$ denote the random set of the \emph{strong} and \emph{weak} devices, respectively.   Let $\koh = \card{\msoh}$ and $\kth = \card{\msth}$ be the discrete random variables of size of those sets, which follows $\koh + \kth = N$. Let $\Roh{j} = \rbo \cdot \ahR_{j}$ denote the transmission rate for device $j \in \msoh$. We assume that $\beta$ is known by all the nodes in the network. Therefore, upon generating the set $\msoh$ for a given realization of the channels, the controller sends the set of indexes in $\msoh$ and the transmission rates $\Roh{j}$, over a control channel to the strong devices. Such information is necessary for the  devices to be able to follow the scheduling order of  transmission in the single-hop phase. The devices in $\msth$ are then scheduled over a two-hop transmission over the remaining $\Tth$, where their messages are aggregated, and  $\alpha \cdot \Tth$  and $(1-\alpha)\cdot \Tth$ are used for broadcasting and relaying the aggregated messages respectively, with $0 \leq \alpha \leq 1$.

In \figref{fig:scheduling} the time scheduling of the proposed \gls{andcoop} transmission protocol is illustrated\footnote{Note that the time dimension is chosen in here as an example. In practice the division of the time-frequency resources between the single-hop and two-hop phases of the protocol can be in either the time domain, frequency domain, or both.}. In the following the proposed \gls{andcoop} transmission protocol is summarized. We assume that the controller has the knowledge of the appropriate $\beta$ and $\alpha$ design parameters, which are acquired off-line  and are shared with the devices (we discuss the optimization of those parameters in the following subsection). A summary of the proposed algorithm in each time cycle is as follows:
\begin{enumerate}
\item 	Using the knowledge of \gls{ap}-device \gls{csi}, the controller finds the set of devices, $\msoh$, that will be scheduled over single-hop rate-adaptive transmission. This is done according to $\msoh = \pmb{S}(\Toh,  \{\gha{i}{j} \})$ in   \eqref{Eq:strongdevices}, while $\Toh = \beta \cdot T$.  
\item 	The controller adapts transmission rate for each device in $\msoh$ according to their instantaneous channel by setting $\Roh{j} = \rbo \cdot \ahR$, thus allocating $B/\Roh{j}$ seconds of the total time for transmission to device $j$. The controller \glspl{ap} will then perform \gls{comp} transmission of the message for each node $j \in \msoh$ with the adapted rate in a \gls{tdma} fashion.
\item 	\label{step-b} All the $B$-bit messages intended for devices in $\msth = \msd \setminus \msoh$ are aggregated  together.  The controller \glspl{ap} jointly broadcast the aggregated message at rate $\Rth{b} = \frac{B \kth}{\alpha \Tth}$,   over the first $\alpha$ portion of  $\Tth$ time\footnote{In practice, the messages can be concatenated before encoding which will potentially increase coding gain and reduce number of decoding attempts for relay devices.}.
\item 	\label{step-receive} All devices in $\msd$ attempt decoding the broadcast message from previous step. The successful devices to decode will act as relays.
\item 	The  \glspl{ap} broadcast the message from step~\ref{step-b} at rate $\Rth{r} = \frac{B \kth}{(1-\alpha) \Tth}$. The relay devices from step~\ref{step-receive} re-encode with the  same code rate, and cooperate as simultaneous relays.
\end{enumerate}

\begin{figure}[t]
\begin{center}
\psfrag{a}[c][c][\scalevalueSS]{$\Toh = \beta \cdot T$}
\psfrag{b}[c][c][\scalevalueSS]{$\alpha \cdot \Tth$}
\psfrag{c}[c][c][\scalevalueSS]{$(1-\alpha)\cdot \Tth$}
\psfrag{t}[c][c][\scalevalueSS]{$T$}
\psfrag{u}[c][c][\scalevalueSS]{$\ldots$}
\psfrag{v}[c][c][\scalevalueSS]{$\Rth{r} = \frac{B \kth}{(1-\alpha) \Tth}$}
\psfrag{w}[c][c][\scalevalueSS]{$\Rth{b} = \frac{B \kth}{\alpha \Tth}$}
\psfrag{x}[c][c][\scalevalueSS]{$\Roh{j_{1}}$}
\psfrag{s}[c][c][\scalevalueSS]{$\Roh{j_{2}}$}
\psfrag{z}[c][c][\scalevalueSS]{$\Roh{j_{\koh}}$}
\includegraphics[width=\scalevalueFigWidth\columnwidth,keepaspectratio]{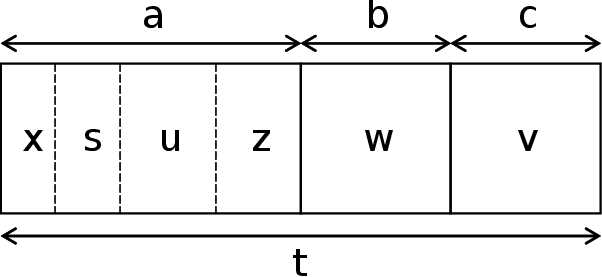}
\caption{Time scheduling illustrated  for the proposed scheme. The transmission rate for each scheduled time slot is identified, assuming $j_1, j_2, \ldots, j_{\koh} \in \msoh$. For the devices in $\msth$, the aggregated message is first broadcast with rate $\Rth{b}$ and then relayed with rate $\Rth{r}$.}
\label{fig:scheduling}
\end{center}
\if\mycmd 1
\vspace{-30pt}
\else
\fi
\end{figure}

\if\mycmd 1
\vspace{-10pt}
\else
\fi
\subsection{A Note on Optimization of the Design}

The proposed transmission protocol can be optimized based on the design parameters $L$, $\beta$ and $\rbo$, to achieve the minimum $\Pout$. We perform numerical optimization in two scenarios of \gls{pcsi} and \gls{icsi}. Note that in \gls{pcsi} scenario, we fix $L = 0$ and $\rbo = 1$, only optimizing with respect to $\beta$, for given transmit power and spectral efficiency. To this end, we fix the value of the time division parameter $\beta$ for all realizations of the channel. Therefore, we are able to numerically optimize $\beta$ for a given setup by running the simulation for all values of $\beta$ from a finite set of real numbers uniformly chosen from the $[0, 1]$ interval to derive $\hat{\beta} = \underset{\beta}{\arg \min}   \Pout(\beta)$. 

Optimization of the \gls{icsi} scenario with exhaustive search across $L$, $\beta$ and $\rbo$, is computationally costly. However, as we will show in the next section, parameter $L$  can be fixed with marginal effect on the outage, which can significantly reduces the search for the optimal $L$ , $\beta$ and $\rbo$.

Parameter $\alpha$ is used to partition the time $\Tth$ between  broadcast and relaying hops of the two-hop cooperative transmission. Increasing $\alpha$ results in smaller $\Rth{b}$ which increases chances of decoding for all nodes and results in larger expected number of relay nodes for the second hop. In turn, it will increase $\Rth{r}$ and decrease chance of decoding in the second hop. 
We note that the optimization of $\alpha$ is not a trivial problem however, following the observation in \cite{swamy2015cow}, the average reliability gain from optimizing $\alpha$ with respect to $\alpha = 0.5$ can be marginal, therefore, in the numerical analysis that will follow, we fix $\alpha = 0.5$. 

\if\mycmd 1
\vspace{-10pt}
\else
\fi
\subsection{Outage Analysis}

We analyze  \emph{system outage probability} of the proposed system, denoted by $\Pout$ and defined as the average probability that at least one device fails to decode its own message at the end of time-cycle $T$.  Let $\Poh(\msoh)$ and $\Pth(\msth)$    denote the  probability of outage for at least one device in, respectively,  the adaptive-rate single-hop phase and the two-hop cooperative phase of the transmission protocol.

%
%
\paragraph*{Single-hop rate-adaptive transmission} Conditioned on \gls{csi}, we have 
\begin{align}\label{Eq:outage-1hop-icsi}
\Poh(\msoh | \{ \hha{i}{j} \}) =\prob{ \exists j \in \msoh,  \rbo \cdot \ahR > \aR },
\end{align}
where expectation is over channel gains.
For the case of \gls{icsi} with channel estimation error, the outage probability in \eqref{Eq:outage-1hop-icsi} is non-zero as also discussed in \cite[Sec. III-A]{Rebal:2018TCOM}. Characterization of $\Poh$ is not analytically tractable. However, when $\beta = 1$, i.e., all the devices are scheduled with single-hop transmission, outage probability is equivalent to the probability of \emph{time overflow}, i.e., the chance that $\Toh$ is too short given the channel gains to successfully accommodate  all the packets. In this case, analytical bounds on the outage probability can be found  in \cite{Rebal:2018TCOM}, where it is shown that in practical ranges of \gls{snr}, as cell load $N \times B$ grows, the chance of time overflow increases above \gls{urllc} target outage probability and dominates system outage regardless of the precision of channel estimation. 

For the case of \gls{pcsi}, when $\beta < 1$,  all the devices in $\msoh$ pass the success condition  in \eqref{Eq:linkoutage}, resulting in $\Poh(\msoh | \{ \ha{i}{j} \}) = 0$. Therefore, the following is valid for \gls{pcsi} scenario.
\begin{align}\label{Eq:outage-1hop}
\Poh(\msoh | \{ \ha{i}{j} \}) = \left\{\begin{matrix}
0 & 0 \leq \beta <1 \\ 
\prob{\sum_{j \in \msd} \frac{1}{\aR_j} > \frac{ \Toh}{B}} & \beta = 1.
\end{matrix}\right.
\end{align}


\paragraph*{Two-hop cooperative transmission}
Conditioned on \gls{csi},   the outage probability of the two-hop phase is given by
\begin{align}
\Pth(\msth \mid \{ \hha{i}{j} \}) =  \prob{ \exists j \in \msth \setminus \msr :   \sum_{k \in \msr}\gd{k}{j}   <  \zeta },
\nonumber
\end{align}
where $\zeta = 2^{\Rth{r}/W} -\sum_{i \in \msa}\ga{i}{j} -1$, and $\msr$ is the set of relay devices, defined as
\begin{align}
\msr = \{j \in \msd:  \Rth{b} \leq \aR_{j} \}.
\end{align}
\if\mycmd 2
\begin{figure*}[b]
\normalsize
\hrulefill
\begin{align}\label{Eq:Ptwohop}
\Pth\left(\msd \right) = \sum_{\substack{ \forall \mss \subset \msd }} & \prob{\forall j \in \mss : \; W \log \bigg( 1 + \sum_{i \in \msa}\ga{i}{j}  \bigg) \geq \Rth{b} } \cdot \prob{\forall j \in \msd \setminus \mss : \; W \log \bigg( 1 + \sum_{i \in \msa}\ga{i}{j}  \bigg) < \Rth{b} } \\
& \cdot \left( 1 - \prob{\forall j \in \msd \setminus \mss : \; W \log \bigg( 1 + \sum_{i \in \msa}\ga{i}{j} + \sum_{k \in \mss}\gd{k}{j}  \bigg) \geq \Rth{r} \; \middle| \; \forall j \in \msd \setminus \mss : \; W \log \bigg( 1 + \sum_{i \in \msa}\ga{i}{j}  \bigg) < \Rth{b} } \right)
\nonumber
\end{align}
\end{figure*}
Although, averaging $\Pth(\msth \mid \{ \ga{i}{j} \})$ over channel gains is not analytically tractable, for the purpose of analyzing the diversity order gain of the transmission protocol, we are interested in the outage probability when all devices are served in  the two-hop cooperative phase. 

An expression for $\Pth$, conditioned on $\kth = N$ can be written as follows in \eqref{Eq:Ptwohop} at the bottom of the page.
\else
Although, averaging $\Pth(\msth \mid \{ \ga{i}{j} \})$ over channel gains is not analytically tractable, for the purpose of analyzing the diversity order gain of the transmission protocol, we are interested in the outage probability when all devices are served in  the two-hop cooperative phase. 

An expression for $\Pth$, conditioned on $\kth = N$ can be written as follows in \eqref{Eq:Ptwohop}.
\fontsize{9}{11}\selectfont
\begin{align}\label{Eq:Ptwohop}
\Pth\left(\msd \right) = \sum_{\substack{ \forall \mss \subset \msd }} & \prob{\forall j \in \mss : \; W \log \bigg( 1 + \sum_{i \in \msa}\ga{i}{j}  \bigg) \geq \Rth{b} } \cdot \prob{\forall j \in \msd \setminus \mss : \; W \log \bigg( 1 + \sum_{i \in \msa}\ga{i}{j}  \bigg) < \Rth{b} } \\
& \cdot \left( 1 - \prob{\forall j \in \msd \setminus \mss : \; W \log \bigg( 1 + \sum_{i \in \msa}\ga{i}{j} + \sum_{k \in \mss}\gd{k}{j}  \bigg) \geq \Rth{r} \; \middle| \; \forall j \in \msd \setminus \mss : \; W \log \bigg( 1 + \sum_{i \in \msa}\ga{i}{j}  \bigg) < \Rth{b} } \right)
\nonumber
\end{align}
\normalsize
\fi
For a practical wireless network in which the channel fading distribution depends on path-loss and shadowing, the evaluation of \eqref{Eq:Ptwohop} is challenging. Simplification of \eqref{Eq:Ptwohop} can be obtained in the following special case. Let's assume a setup with fixed nominal \gls{snr} $\rho$ on all \gls{ap}-device and device-device links with \gls{iid} small-scale fading $h$. 
In such a scenario, using \eqref{Eq:linkoutage}, the probability of decoding failure with $m$ cooperating transmitters at rate $R$  is as follows.
\begin{align}\label{Eq:pfail}
\Pm{m,R} = \prob{W \log\Big( 1 + \rho \sum_{l = 1}^{m}|h_l|^2 \Big) < R},
\end{align}
where $h_l \sim CN(0,1)$, are \gls{iid} random variables. Therefore, $\sum_{l = 1}^{m}|h_l|^2$ has the Erlang distribution and $p^m$ can be computed as \cite{Laneman:2003}
\begin{align}
\Pm{m,R}  = \gamma\left( m,\frac{\omega}{\Pt} \right),
\end{align}
where $\gamma(x,y) = \int_{0}^{y} t^{x-1} e^{-t} dt$ is the incomplete Gamma function. Moreover, $\omega = W \cdot \sigma_0 \cdot (2^{R/W} - 1)$ and $\sigma_0$ denotes \gls{psd} of the \gls{awgn} and $\Pt$ denotes the transmit power at each transmitting node.

%
%
%
In such simplified scenario, as similarly suggested in \cite{swamy2015cow}, \eqref{Eq:Ptwohop} can be reformulated as follows.
\if\mycmd 1
\begin{align}\label{Eq:Ptwohop-4}
\Pth\left(\msd \right) = \sum_{n = 0}^{N-1}\left( q_{b}^{M}\right)^{(N-n)} \left(1-q_{b}^M \right)^{n}  \binom{N}{n} \left(1 - \left( 1 - {q_{r}^{(M+n)}} \right)^{(N-n)} \right),
\end{align}
\else
\begin{align}\label{Eq:Ptwohop-4}
&\Pth\left(\msd \right) = \\
&\sum_{n = 0}^{N-1}\left( q_{b}^{M}\right)^{(N-n)} \left(1-q_{b}^M \right)^{n}  \binom{N}{n} \left(1 - \left( 1 - {q_{r}^{(M+n)}} \right)^{(N-n)} \right),
\nonumber
\end{align}
\fi  
where  $q_{b}^M = \Pm{M,\Rth{b}}$ and 
\begin{align}\label{Eq:qr}
q_{r}^{(M+n)} = \min \{1, \Pm{M+n,\Rth{r}} / \Pm{M,\Rth{b}} \},
\end{align}
is the conditional failure probability  of a device in relaying hop given 
\if\mycmd 1
\else
that 
\fi
it failed in broadcast hop.

\if\mycmd 1
\vspace{-10pt}
\else
\fi
\subsection{Diversity-Multiplexing Tradeoff Analysis}
\label{Sec:DMT}

In this section, we analyze the diversity-multiplexing trade-off of the proposed scheme for the case of \gls{pcsi}, assuming independent Rayleigh distributed fading for all links, using the definition in \eqref{Eq:dive_order_def}. First, note that when transmit power goes to $\infty$, the achievable transmission rate from \eqref{Eq:achievablerate} also approaches $\infty$, resulting in $ \msoh = \msd$, when $0 < \beta \leq 1$. Therefore, unless $\beta = 0$, all the devices are scheduled in the single-hop rate-adaptive phase. For that reason, we  analyze the diversity-multiplexing tradeoff for two extreme cases of $\beta = 1$ and $\beta = 0$, respectively, when all the devices are scheduled in single-hop rate-adaptive phase and, when all the devices are scheduled in two-hop cooperative phase. This also provides lower bounds for the diversity-multiplexing tradeoff of the proposed \gls{andcoop} scheme.

\begin{proposition}
\label{Prop:1h}
When all the devices are scheduled in single-hop rate-adaptive phase, the diversity order of the considered system at zero multiplexing is given by
\begin{align}
\label{Eq:div_order_1hop}
d_\textit{single-hop}(0) = M.
\end{align}
\end{proposition}
\begin{proof}
See \appref{App:A}.
\end{proof}
Moreover, following the proof in \appref{App:A}, upper and lower bounds of the diversity-multiplexing tradeoff are readily derived as follows for the case where all the devices are scheduled over single-hop transmission.
\begin{align}\label{Eq:divmux_1hop}
M (1-r) \leq d_\textit{single-hop}(r) \leq  M (1-\frac{r}{N}).
\end{align}

For the case when all the devices are scheduled in two-hop cooperative phase (i.e., the OccupyCoW protocol in \cite{swamy2015cow}), the following proposition presents the diversity-multiplexing tradeoff  for the general case of $0 < \alpha <1$, when $M>1$.

\if\mycmd 1
\vspace{-10pt}
\else
\fi
\begin{proposition}
\label{Prop:2h}
When  all the devices are scheduled in two-hop cooperative phase, the diversity-multiplexing tradeoff is given by
\begin{align}\label{Eq:DMT-2hop}
d_\textit{two-hop}(r) = (M+N-1)\left( 1- \frac{r}{1-\alpha} \right),
\end{align}
where at zero multiplexing gain it yields
\begin{align}
\label{Eq:div_orde_2hop}
d_\textit{two-hop}(0) = M+N-1.
\end{align}
\end{proposition}
\begin{proof}
See \appref{Proof_prop_2h}.
\end{proof}
Intuitively, diversity order $M + N -1$, i.e.,  the slope of outage probability curve in high \gls{snr}, corresponds to the case when all the $M$ \glspl{ap} cooperate with  $N-1$ strongest devices to transmit to the weakest device. It is  an interesting observation from \eqref{Eq:DMT-2hop} that the tradeoff between diversity order $d$ and multiplexing gain $r$ in high \gls{snr} is controlled by the inverse of $1- \alpha$. In other words, the duration of the relaying phase, which directly affects the outage probability for the weakest devices, controls the diversity multiplexing tradeoff in high \gls{snr}. Therefore, for a given multiplexing gain $r$, the maximum diversity is achieved when $\alpha$ is the minimum allowed value for \eqref{Eq:DMT-2hop} to be non-negative, i.e., $0 < \alpha \leq 1-r$. Since $\alpha$ cannot be zero in \eqref{Eq:DMT-2hop}, the following upper bound is valid for the two-hop transmission for $r > 0$.
\begin{align}
\label{Eq:div_orde_2hop_up}
d_\textit{two-hop}(r) < (M+N-1) ( 1- r ).
\end{align}
This upper bound is depicted by the dashed green line in \figref{fig:div_mux_tradeoff}.  For the  case of $\alpha = \frac{1}{2}$, which is the exercised design in \cite{Arvin:2019,swamy2015cow,Kaiming:2019ratespliting}, the diversity-multiplexing tradeoff in  \eqref{Eq:DMT-2hop} becomes 
\begin{align}\label{Eq:2hop_divmux_half}
d_\textit{two-hop}(r)|_{\alpha=1/2} = (M+N-1) ( 1- 0.5 \cdot r )
\end{align}
As also depicted in \figref{fig:div_mux_tradeoff}, in such case the maximum achievable  multiplexing gain is $r = 0.5$, intuitively corresponding to the fact that each message is transmitted twice with the same rate, once in broadcast phase and once more in the relaying phase. 

It is clear that the  two-hop operation in \eqref{Eq:2hop_divmux_half} is spectrally inefficient, due to the fact that at very low system outage probability  the gain from cooperative relaying is small, thus only a small percentage of the devices will benefit from relaying. On the other hand, the single-hop transmission from \propref{Prop:1h} can achieve higher multiplexing gain than \eqref{Eq:2hop_divmux_half}. This further suggests that it is best to capture the relaying benefit only for that small percentage of the devices that experience a \emph{weak} channel to the \glspl{ap}, while  enjoying the high  rate single-hop transmission to the devices with \emph{strong} channel, which in turn increases the overall spectral efficiency. Moreover, note that since by design the outage probability of the proposed \gls{andcoop} scheme is upper bounded by that of the two extremes studied in \propref{Prop:1h} and \propref{Prop:2h}, then from \eqref{Eq:dive_order_def} we readily have $d_\textit{single-hop}(r) \leq d_\textit{\gls{andcoop}}(r)$ and  $d_\textit{two-hop}(r) \leq d_\textit{\gls{andcoop}}(r)$ for given $\alpha$ and $r$.

%
%

\begin{figure}[t]
\begin{center}
\psfrag{xlabel}[c][c][\scalevalueS]{multiplexing gain, $r$}
\psfrag{ylabel}[c][c][\scalevalueS]{diversity order, $d(r)$}
\psfrag{x1}[lc][lc][\scalevalueSS]{\propref{Prop:andcoop}}
\psfrag{XXXXXXXXXx1}[lc][lc][\scalevalueSS]{two-hop; up}
\psfrag{x2}[lc][lc][\scalevalueSS]{two-hop; $\alpha = \frac{1}{2}$}
\psfrag{x3}[lc][lc][\scalevalueSS]{single-hop; low}
\psfrag{x4}[lc][lc][\scalevalueSS]{single-hop; up}
\psfrag{y1}[rc][rc][\scalevalueSS]{$M$}
\psfrag{ye}[rc][rc][\scalevalueSS]{$M+N-1$}
\psfrag{xe}[rc][rc][\scalevalueSS]{$\frac{N}{N+1}$}
\if\mycmd 1
\includegraphics[width=.32\columnwidth,keepaspectratio]{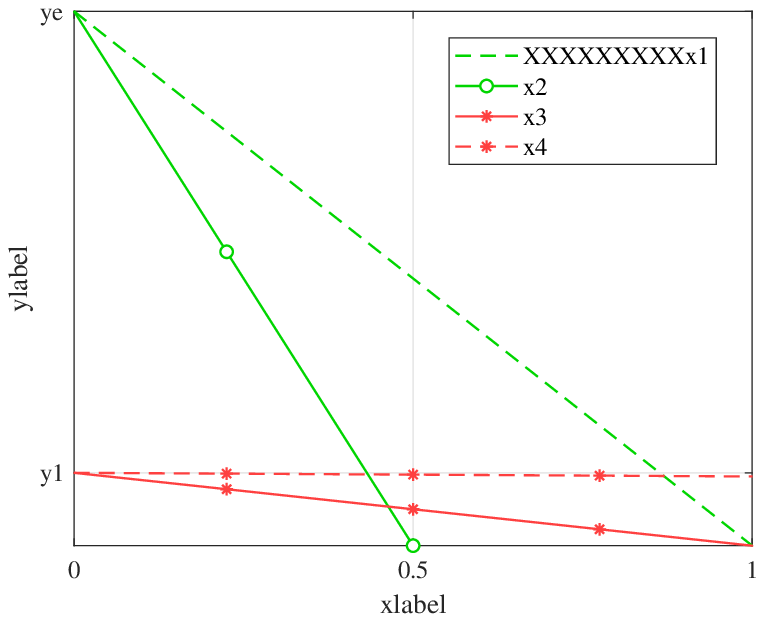}
\else
\includegraphics[width=0.82\columnwidth,keepaspectratio]{div_mux_tradeoff.eps}
\fi
\caption{Diversity-multiplexing tradeoff.}
\label{fig:div_mux_tradeoff}
\end{center}
\if\mycmd 1
\vspace{-30pt}
\else
\fi
\end{figure}
\section{Numerical Results}
\label{Sec:Results}

This section  presents numerical results from simulating \gls{dl} operation  of the network described in \secref{Sec:ProblemSetup}, where $M$ single-antenna \glspl{ap} and $N$ single-antenna devices are randomly distributed across the factory area. To this end, we start by analyzing the performance with the assumption of \gls{pcsi}, to focus on system parameter designs. Then, we extend the analysis to the case of \gls{icsi}, and study the impact of \gls{csi} estimation error on the performance indicators.

The plots presented throughout this section are generated using a system-level simulation, where we adopt the  parameters summarized in \tabref{Tab:Parameters}  (except in cases where stated otherwise). Path loss exponents are determined using the factory and open-plan building channel model in \cite{rapp91}. The probability that a link is \gls{los} is derived based on the distance of the communicating  nodes, $\nu$, using the following function.
\begin{align}
\plos(\nu) = a + \indic{\nu \leq b}\,\frac{1-a}{b^2}\,(\nu-b)^2,
\end{align}
where $a$ is a fixed probability mass, and $b$ is the cutoff above which the probability of a link being \gls{los} becomes fixed to $a$. The link is therefore \gls{nlos} with probability of $1-\plos(\nu)$.  The cycle duration, number of devices, data per device and bandwidth are given values similar to those in \cite{swamy2015cow} and \cite{weiner2014design}.

With respect to the optimization of $\beta$, the following  schemes are studied in this section.
\begin{itemize}
\item 	$\beta = \hat{\beta}$: This represents our proposed \gls{andcoop} scheme, where an optimal $\hat{\beta}$ portion of the total time is  allocated for rate-adaptive single-hop transmission to a subset of the devices with the highest instantaneous channel quality. The remaining $1 - \hat{\beta}$ portion of the time is then used for two-hop transmission to the remaining devices.
\item 	$\beta = 1$, ideal rate adaptation: This transmission scheme mimics the typical transmission in cellular technologies, such as \gls{lte}, where, the transmission rate for each user is adapted to the instantaneous channel quality.
\item 	$\beta = 0$, two-hop transmission: This is a special case of our proposed scheme, where the total time resources is allocated to the two-hop transmission towards all the users. The scheme was originally proposed in  \cite{swamy2015cow} and is known as the OccupyCoW protocol.
\end{itemize}


\input{table2.tex}

\if\mycmd 1
\vspace{-10pt}
\else
\fi
\subsection{Performance Analysis with Perfect \gls{csi}}

Let's start by assuming that the controller has perfect knowledge of the channel gains for all the \gls{ap}-device links. We emphasize that in the present work, the \gls{csi} is solely used for the purpose of transmission rate adaptation, meaning that the channel coefficient phase is not collected nor utilized. This differentiates our work from \gls{csi}-based distributed multi-antenna systems, such as in \cite{alonzo2020urllc}, which rely on coherent joint transmission and beam-forming at the transmitter. This further relaxes the assumption of tight synchronization among antennas for distributed cooperative transmission. 

\if\mycmd 1
\begin{figure*}[t]
\centering
\begin{minipage}[b]{.47\textwidth}
\begin{center}
\psfrag{xlabel}[c][c][\scalevalueS]{spectral efficiency, $\eta$  (\gls{bpcu})}
\psfrag{ylabel}[c][c][\scalevalueS]{probability of system outage, $\Pout$}
\psfrag{XXX1}[lc][lc][\scalevalueSS]{$\beta = \hat{\beta}$}
\psfrag{x2}[lc][lc][\scalevalueSS]{$\beta = 0$}
\psfrag{x3}[lc][lc][\scalevalueSS]{$\beta = 1$}
\psfrag{XXX4}[lc][lc][\scalevalueSS]{$M = 1$}
\psfrag{x5}[lc][lc][\scalevalueSS]{$M = 2$}
\psfrag{x6}[lc][lc][\scalevalueSS]{$M = 3$}
\includegraphics[width=\scalevalueFigWidthSmall\columnwidth,keepaspectratio]{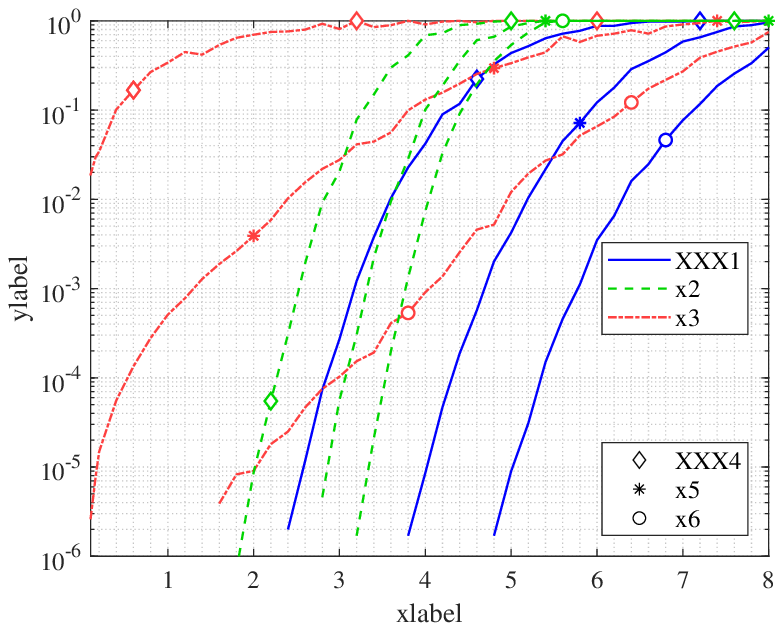}
\caption{System outage probability against spectral efficiency at $\Pa = \Pd = 23$ dBm.}
\label{fig:comp_id_rate}
\end{center}
\end{minipage}\qquad
\begin{minipage}[b]{.47\textwidth}
\begin{center}
\psfrag{xlabel}[c][c][\scalevalueS]{number of weak devices, $\kth$}
\psfrag{ylabel}[c][c][\scalevalueS]{cumulative distribution function}
\psfrag{XXX1}[lc][lc][\scalevalueSS]{$M = 1$}
\psfrag{x2}[lc][lc][\scalevalueSS]{$M = 2$}
\psfrag{x3}[lc][lc][\scalevalueSS]{$M = 3$}
\includegraphics[width=\scalevalueFigWidthSmall\columnwidth,keepaspectratio]{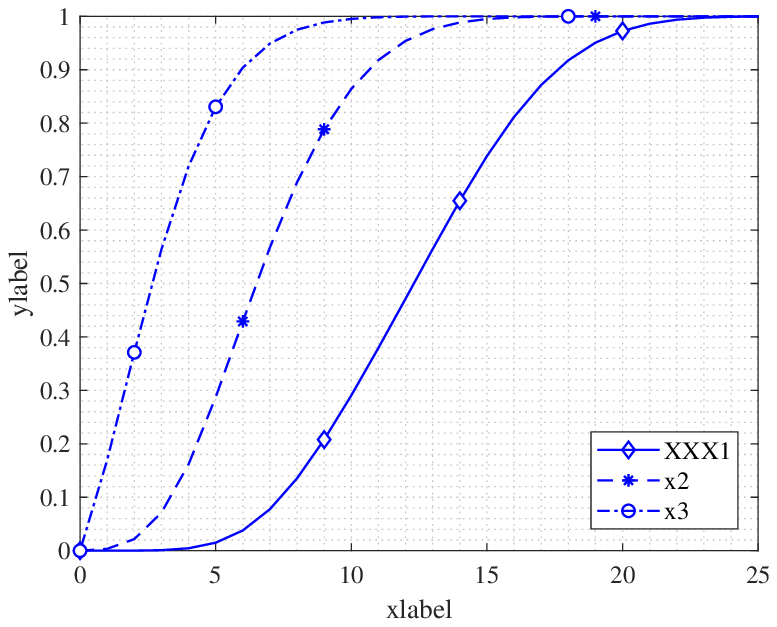}
\caption{Statistics of $\kth$ 
 with $N = 50$, $\Pa = \Pd = 23$ dBm at $\Pout = 10^{-5}$ points of \figref{fig:comp_id_rate}.}\label{fig:weak_node_stats_id}
\end{center}
\end{minipage}
\if\mycmd 1
\vspace{-35pt}
\else
\fi
\end{figure*}
\fi

\if\mycmd 2
\begin{figure}[t]
\begin{center}
\psfrag{xlabel}[c][c][\scalevalueS]{spectral efficiency, $\eta$  (\gls{bpcu})}
\psfrag{ylabel}[c][c][\scalevalueS]{probability of system outage, $\Pout$}
\psfrag{XXX1}[lc][lc][\scalevalueSS]{$\beta = \hat{\beta}$}
\psfrag{x2}[lc][lc][\scalevalueSS]{$\beta = 0$}
\psfrag{x3}[lc][lc][\scalevalueSS]{$\beta = 1$}
\psfrag{XXX4}[lc][lc][\scalevalueSS]{$M = 1$}
\psfrag{x5}[lc][lc][\scalevalueSS]{$M = 2$}
\psfrag{x6}[lc][lc][\scalevalueSS]{$M = 3$}
\includegraphics[width=\scalevalueFigWidthSmall\columnwidth,keepaspectratio]{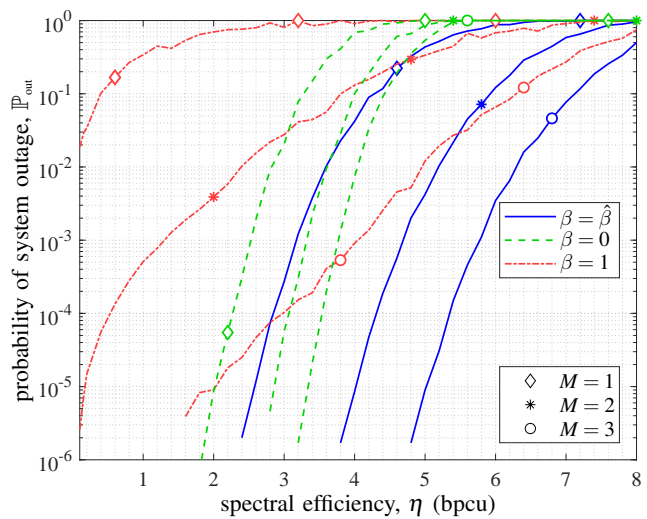}
\caption{System outage probability against spectral efficiency at $\Pa = \Pd = 23$ dBm.}
\label{fig:comp_id_rate}
\end{center}
\if\mycmd 1
\vspace{-35pt}
\else
\fi
\end{figure}
\fi

\paragraph*{Proposed scheme improves spectral efficiency} In \figref{fig:comp_id_rate}, system outage probability is shown against spectral efficiency $\eta$, derived as $\eta = \frac{N B}{T W}$. The proposed \gls{andcoop} scheme with optimized time division of $\beta = \hat{\beta}$, improves spectral efficiency by at least 0.5 \gls{bpcu} when $M = 1$ \gls{ap} is deployed. The gain is higher when a larger number of \glspl{ap} are deployed. Namely, with $M = 2$ and $M = 3$, more than 1 and 1.5 \gls{bpcu} increase in spectral efficiency is achieved with respect to the case of $\beta = 0$, i.e., when only two-hop cooperative transmission is deployed. 

\if\mycmd 2
\begin{figure}[t]
\begin{center}
\psfrag{xlabel}[c][c][\scalevalueS]{number of weak devices, $\kth$}
\psfrag{ylabel}[c][c][\scalevalueS]{cumulative distribution function}
\psfrag{XXX1}[lc][lc][\scalevalueSS]{$M = 1$}
\psfrag{x2}[lc][lc][\scalevalueSS]{$M = 2$}
\psfrag{x3}[lc][lc][\scalevalueSS]{$M = 3$}
\includegraphics[width=\scalevalueFigWidthSmall\columnwidth,keepaspectratio]{weak_node_stats_id.eps}
\caption{Statistics of $\kth$ 
 with $N = 50$, $\Pa = \Pd = 23$ dBm at $\Pout = 10^{-5}$ points of \figref{fig:comp_id_rate}.}
\label{fig:weak_node_stats_id}
\end{center}
\if\mycmd 1
\vspace{-35pt}
\else
\fi
\end{figure}
\fi

The gain in spectral efficiency is thanks to transmitting packets to \emph{strong} devices with high rate in the single-hop phase, allowing the two-hop phase to accommodate the \emph{weak} devices reliably at even  large packet sizes. This way, the robustness of rate-adaptation to increase in load is combined together with the robustness of OccupyCoW to fading, providing an improved reliability at a higher spectral efficiency. \figref{fig:weak_node_stats_id} shows the \gls{cdf} of the number of weak devices scheduled with two-hop transmission by the proposed \gls{andcoop} scheme with optimized $\beta$. The statistics are collected for the points in \figref{fig:comp_id_rate} where the proposed \gls{andcoop} achieves $\Pout = 10^{-5}$. Interestingly, the average number of users scheduled with two-hop transmission is just below 13, 7 and 3 respectively for the case of 1, 2 and 3 \glspl{ap}.

\if\mycmd 2
\begin{figure}[t]
\begin{center}
\psfrag{xlabel}[c][c][\scalevalueS]{transmit power (dBm)}
\psfrag{ylabel}[c][c][\scalevalueS]{probability of system outage, $\Pout$}
\psfrag{XXX1}[lc][lc][\scalevalueSS]{$\beta = \hat{\beta}$}
\psfrag{x2}[lc][lc][\scalevalueSS]{$\beta = 0$}
\psfrag{x3}[lc][lc][\scalevalueSS]{$\beta = 1$}
\psfrag{XXX4}[lc][lc][\scalevalueSS]{$M = 1$}
\psfrag{x5}[lc][lc][\scalevalueSS]{$M = 2$}
\psfrag{x6}[lc][lc][\scalevalueSS]{$M = 3$}
\includegraphics[width=\scalevalueFigWidthSmall\columnwidth,keepaspectratio]{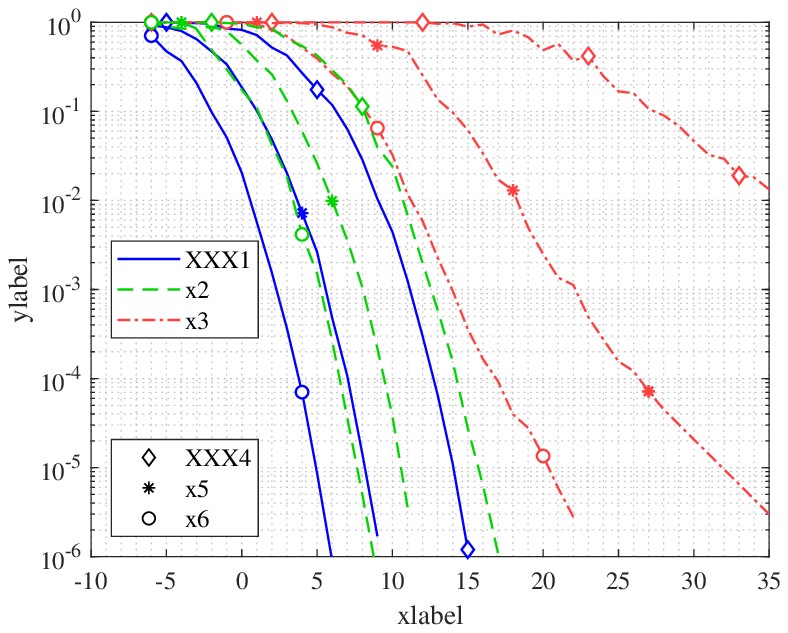}
\caption{System outage probability against transmit power, $\Pa = \Pd$, for $B = 50$ bytes per device.}
\label{fig:comp_id_power}
\end{center}
\if\mycmd 1
\vspace{-35pt}
\else
\fi
\end{figure}
\fi

\paragraph*{Ultra-reliability at lower transmit power} By fully exploiting the diversity gain in the network, the proposed \gls{andcoop} scheme relaxes the need for high \gls{snr} to achieve ultra-reliability. As it is shown in \figref{fig:comp_id_power}, system outage probability of $\Pout = 10^{-5}$, the required transmit power of the proposed scheme reduces by a few dB compared to the OccupyCoW protocol. Such transmit power gap, when compared against the case of single-hop with ideal rate adaptation, can grow to tens of dB. Aside from improving the overall  energy efficiency of the system, operating at a lower transmit power can also reduce the interference generated by the cell towards neighbouring cells, in case of a multi-cell operation, as also studied in \cite{Arvin:2019}. Moreover, the proposed \gls{andcoop} can naturally reduce the average relaying time per relay device, by reducing the overall duration of the relaying phase. The combined effect is a significant reduction in the average consumed energy in the  relaying phase across devices, as depicted in \figref{fig:relay_power_consumption_stats_id_jules}. The curves suggest $30\%$ to $40\%$ reduction in average relaying energy consumption in all cases.

\if\mycmd 2
\begin{figure}[t]
\begin{center}
\psfrag{xlabel}[c][c][\scalevalueS]{average transmit energy per device ($\mu$J)}
\psfrag{ylabel}[c][c][\scalevalueS]{cumulative distribution function}
\psfrag{XXX1}[lc][lc][\scalevalueSS]{$\beta = \hat{\beta}$}
\psfrag{x3}[lc][lc][\scalevalueSS]{$\beta = 0$}
\psfrag{XXX4}[lc][lc][\scalevalueSS]{$M = 1$}
\psfrag{x5}[lc][lc][\scalevalueSS]{$M = 2$}
\psfrag{x6}[lc][lc][\scalevalueSS]{$M = 3$}
\includegraphics[width=\scalevalueFigWidthSmall\columnwidth,keepaspectratio]{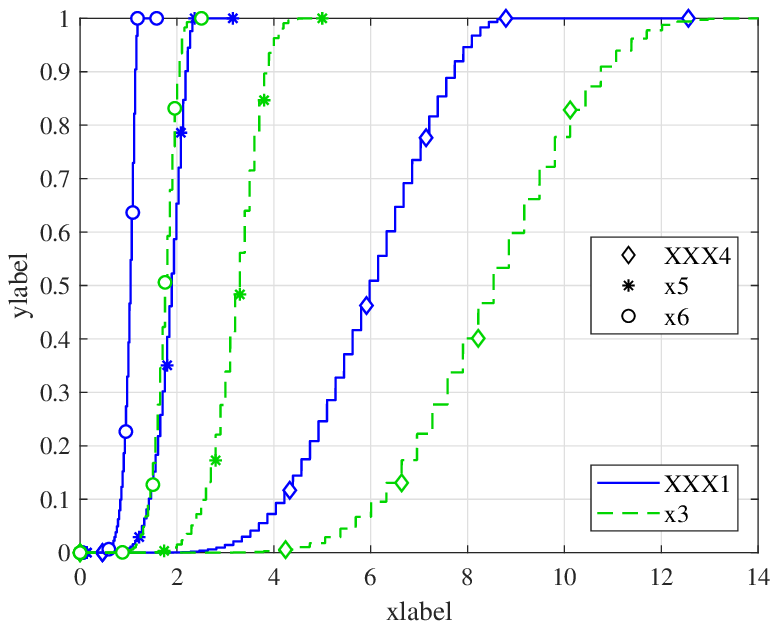}
\caption{Statistics of  average consumed energy per device for relaying with $N = 50$, $B = 50$ bytes at $\Pout = 10^{-5}$ points of \figref{fig:comp_id_power}.}
\label{fig:relay_power_consumption_stats_id_jules}
\end{center}
\if\mycmd 1
\vspace{-35pt}
\else
\fi
\end{figure}
\fi

\if\mycmd 1
\begin{figure*}
\centering
\begin{minipage}[b]{.47\textwidth}
\begin{center}
\psfrag{xlabel}[c][c][\scalevalueS]{transmit power (dBm)}
\psfrag{ylabel}[c][c][\scalevalueS]{probability of system outage, $\Pout$}
\psfrag{XXX1}[lc][lc][\scalevalueSS]{$\beta = \hat{\beta}$}
\psfrag{x2}[lc][lc][\scalevalueSS]{$\beta = 0$}
\psfrag{x3}[lc][lc][\scalevalueSS]{$\beta = 1$}
\psfrag{XXX4}[lc][lc][\scalevalueSS]{$M = 1$}
\psfrag{x5}[lc][lc][\scalevalueSS]{$M = 2$}
\psfrag{x6}[lc][lc][\scalevalueSS]{$M = 3$}
\includegraphics[width=\scalevalueFigWidthSmall\columnwidth,keepaspectratio]{comp_id_power.eps}
\caption{System outage probability against transmit power, $\Pa = \Pd$, for $B = 50$ bytes per device.}
\label{fig:comp_id_power}
\end{center}
\end{minipage}\qquad
\begin{minipage}[b]{.47\textwidth}
\begin{center}
\psfrag{xlabel}[c][c][\scalevalueS]{average transmit energy per device ($\mu$J)}
\psfrag{ylabel}[c][c][\scalevalueS]{cumulative distribution function}
\psfrag{XXX1}[lc][lc][\scalevalueSS]{$\beta = \hat{\beta}$}
\psfrag{x3}[lc][lc][\scalevalueSS]{$\beta = 0$}
\psfrag{XXX4}[lc][lc][\scalevalueSS]{$M = 1$}
\psfrag{x5}[lc][lc][\scalevalueSS]{$M = 2$}
\psfrag{x6}[lc][lc][\scalevalueSS]{$M = 3$}
\includegraphics[width=\scalevalueFigWidthSmall\columnwidth,keepaspectratio]{relay_power_consumption_stats_id_jules.eps}
\caption{Statistics of  average consumed energy per device for relaying with $N = 50$, $B = 50$ bytes at $\Pout = 10^{-5}$ points of \figref{fig:comp_id_power}.}
\label{fig:relay_power_consumption_stats_id_jules}
\end{center}
\end{minipage}
\if\mycmd 1
\vspace{-35pt}
\else
\fi
\end{figure*}
\fi

\if\mycmd 2
\begin{figure}[t]
\begin{center}
\psfrag{xlabel}[c][c][\scalevalueS]{probability of system outage, $\Pout$}
\psfrag{ylabel}[c][c][\scalevalueS]{empirical outage exponent}
\psfrag{XXX1}[lc][lc][\scalevalueSS]{$\beta = \hat{\beta}$}
\psfrag{x2}[lc][lc][\scalevalueSS]{$\beta = 0$}
\psfrag{x3}[lc][lc][\scalevalueSS]{$\beta = 1$}
\psfrag{XXX4}[lc][lc][\scalevalueSS]{$M = 1$}
\psfrag{x5}[lc][lc][\scalevalueSS]{$M = 2$}
\psfrag{x6}[lc][lc][\scalevalueSS]{$M = 3$}
\includegraphics[width=\scalevalueFigWidthSmall\columnwidth,keepaspectratio]{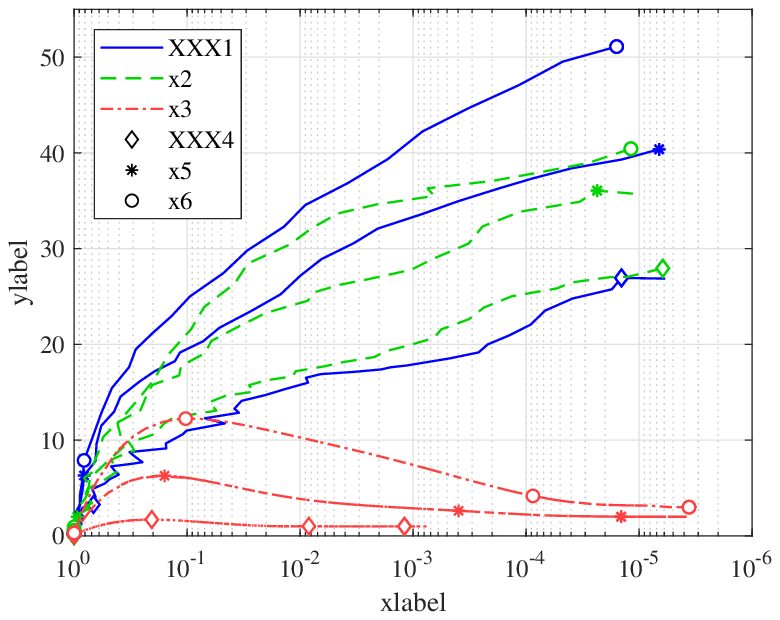}
\caption{Empirical diversity order in finite \gls{snr} against system outage probability for $B = 50$ bytes per device.}
\label{fig:div_order}
\end{center}
\if\mycmd 1
\vspace{-35pt}
\else
\fi
\end{figure}
\fi

\paragraph*{Quick reach to the maximum diversity gain} We study the empirical outage exponent of the proposed transmission scheme in \figref{fig:div_order}. For that purpose, we adopt a simulation setup where all links exhibit a single nominal average \gls{snr} value, i.e., removing the effect of large-scale fading in channel gain. The system outage probability is then simulated across a finite range of link \gls{snr}. As depicted in \figref{fig:div_order},  with optimal $\beta = \hat{\beta}$, the proposed protocol  reaches quickly to the maximum achievable diversity order of $M+N-1$ at $M = 3$, confirming the derivation in \eqref{Eq:div_orde_2hop}. This is thanks to collecting the multi-user diversity gain at its best, by rate-adaptive  scheduling of the devices with strong channel condition, while exploiting the spatial diversity gain from cooperative relaying towards the devices with poor channel condition. Interestingly, with $\beta  = 1$, the gain from multi-user diversity can initially increase the outage exponent. By increasing \gls{snr}, where all devices will naturally be scheduled with the same transmission rate, the diversity order approaches to $M$, as was also suggested by the  derivation in \eqref{Eq:div_order_1hop}.

\if\mycmd 2
\begin{figure}[t]
\begin{center}
\psfrag{ylabel}[c][c][\scalevalueS]{probability of system outage, $\Pout$}
\psfrag{xlabel}[c][c][\scalevalueS]{number of devices, $N$}
\psfrag{XXX1}[lc][lc][\scalevalueSS]{$\beta = \hat{\beta}$}
\psfrag{x2}[lc][lc][\scalevalueSS]{$\beta = 0$}
\psfrag{x3}[lc][lc][\scalevalueSS]{$\beta = 1$}
\psfrag{XXX4}[lc][lc][\scalevalueSS]{$M = 1$}
\psfrag{x5}[lc][lc][\scalevalueSS]{$M = 3$}
\includegraphics[width=\scalevalueFigWidthSmall\columnwidth,keepaspectratio]{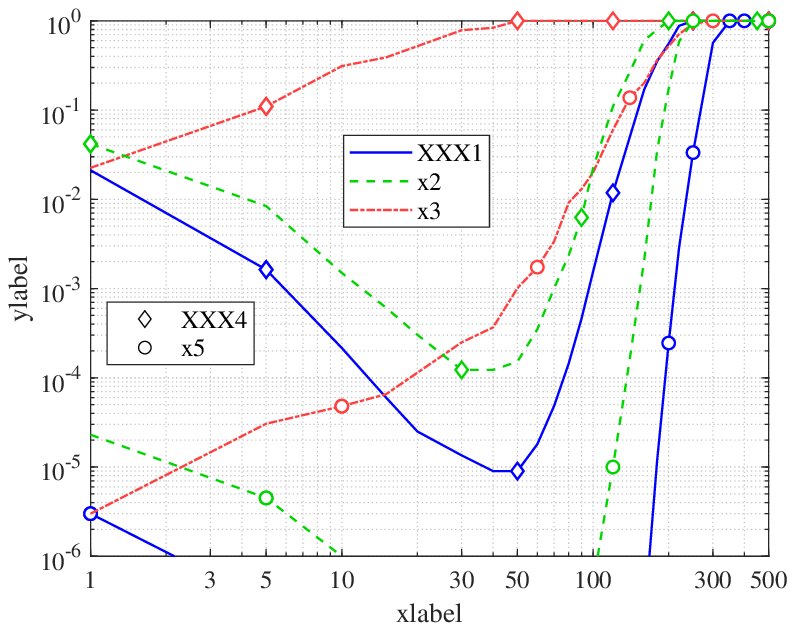}
\caption{System outage probability against number of devices, $N$, for $\Pa = \Pd = 14$ dBm and $B = 50$ bytes per device.}
\label{fig:comp_id_outage_population_pwr_14}
\end{center}
\if\mycmd 1
\vspace{-35pt}
\else
\fi
\end{figure}
\fi

\if\mycmd 1
\begin{figure*}
\centering
\begin{minipage}[b]{.47\textwidth}
\begin{center}
\psfrag{xlabel}[c][c][\scalevalueS]{probability of system outage, $\Pout$}
\psfrag{ylabel}[c][c][\scalevalueS]{empirical outage exponent}
\psfrag{XXX1}[lc][lc][\scalevalueSS]{$\beta = \hat{\beta}$}
\psfrag{x2}[lc][lc][\scalevalueSS]{$\beta = 0$}
\psfrag{x3}[lc][lc][\scalevalueSS]{$\beta = 1$}
\psfrag{XXX4}[lc][lc][\scalevalueSS]{$M = 1$}
\psfrag{x5}[lc][lc][\scalevalueSS]{$M = 2$}
\psfrag{x6}[lc][lc][\scalevalueSS]{$M = 3$}
\includegraphics[width=\scalevalueFigWidthSmall\columnwidth,keepaspectratio]{div_order_snr.eps}
\caption{Empirical diversity order in finite \gls{snr} against system outage probability for $B = 50$ bytes per device.}
\label{fig:div_order}
\end{center}
\end{minipage}\qquad
\begin{minipage}[b]{.47\textwidth}
\begin{center}
\psfrag{ylabel}[c][c][\scalevalueS]{probability of system outage, $\Pout$}
\psfrag{xlabel}[c][c][\scalevalueS]{number of devices, $N$}
\psfrag{XXX1}[lc][lc][\scalevalueSS]{$\beta = \hat{\beta}$}
\psfrag{x2}[lc][lc][\scalevalueSS]{$\beta = 0$}
\psfrag{x3}[lc][lc][\scalevalueSS]{$\beta = 1$}
\psfrag{XXX4}[lc][lc][\scalevalueSS]{$M = 1$}
\psfrag{x5}[lc][lc][\scalevalueSS]{$M = 3$}
\includegraphics[width=\scalevalueFigWidthSmall\columnwidth,keepaspectratio]{comp_id_outage_population_pwr_14.eps}
\caption{System outage probability against number of devices, $N$, for $\Pa = \Pd = 14$ dBm and $B = 50$ bytes per device.}
\label{fig:comp_id_outage_population_pwr_14}
\end{center}
\end{minipage}
\if\mycmd 1
\vspace{-35pt}
\else
\fi
\end{figure*}
\fi

%

\paragraph*{Improved scalability with network size} In practical industrial networks, the number of devices connected to the same controller can become largely dynamic. Therefore, it is crucial for the transmission protocol to be able to scale with the network size, up or down, without depriving it of reliability. To that end, in \figref{fig:comp_id_outage_population_pwr_14} the system outage probability of the three transmission schemes are tested against a range of network sizes, i.e., the number of devices in the network, $N$. We fixed the packet size for each device to $B = 50$ bytes to imitate the realistic conditions. 

The system outage probability for the case of single-hop transmission with ideal rate adaptation increases, almost linearly, by increasing the network size. This is an expected outcome since the system is unable to gain from the increase in number of devices (maximum diversity order or $M$, as proposed in \eqref{Eq:div_order_1hop}). Thus, increasing the number of devices, merely translates into a higher likelihood of scheduling  time-overflow.

On the contrary, the proposed \gls{andcoop} protocol and the OccupyCoW protocol can benefit from the increase in network device. In fact, by increasing the network size, the potential cooperative diversity gain also increases, which in turn reduces the system outage probability. Meanwhile, by increasing the network size, the average transmission rate increases too, which has an opposite effect on outage probability. Therefore, for those two schemes, we observe a turning point for system outage probability. Overall, the proposed \gls{andcoop} can guarantee $\Pout \leq 10^{-5}$ for $N \leq 180$ with $M = 3$ \glspl{ap} in this example. For the OccupyCoW scheme this reduces to only $2<N\leq 120$. It should be noted that at $N = 1$,  the proposed scheme is equivalent to single-hop transmission since the total resources are allotted to the single device. Moreover, the overall scheduling overhead  also increases with the number of devices increasing. However, since a fixed packet size per device is assumed in this analysis, by introducing a fixed scheduling overhead per device, the trend of the curves in \figref{fig:comp_id_outage_population_pwr_14} and the above conclusions will remain intact.


\if\mycmd 1
\vspace{-10pt}
\else
\fi
\subsection{Effect of Imperfect \gls{csi}}

As discussed earlier in \secref{Sec:transmissionprotocol}, due to the inevitable \gls{csi} estimation error, the transmitter must use a back-off parameter $0 < \rbo \leq 1$ to adjust the transmission rate that is adapted to the \gls{icsi}.  
Moreover, the effect of \gls{icsi} is only on the single-hop rate-adaptive transmission phase. The two-hop  phase, thanks to the fixed-rate transmission, encounters no impact from the \gls{icsi}.

Numerical optimization in presence of \gls{csi} error is a demanding task which requires exhaustive search for the optimal operating point across three parameters $L$, $\beta$ and $\rbo$. Therefore, it is crucial to understand the impact that each of those parameters have on the performance in order to reduce the optimization complexity. 
\if\mycmd 1
\begin{figure*}
\centering
\begin{minipage}[b]{.47\textwidth}
\begin{center}
\psfrag{xlabel}[c][c][\scalevalueS]{number of pilot symbols per device, $L$}
\psfrag{ylabel}[c][c][\scalevalueS]{system outage probability, $\Pout$}
\psfrag{XXXXx1}[lc][lc][\scalevalueSS]{best case}
\includegraphics[width=\scalevalueFigWidthSmall\columnwidth,keepaspectratio]{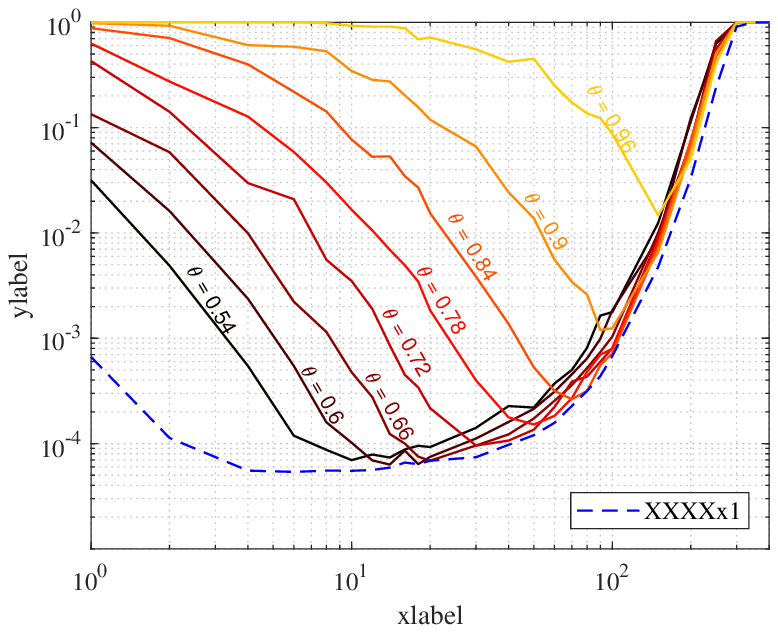}
\caption{System outage probability of the proposed \gls{andcoop} for different  $\rbo$ values. We fixed $\Pa = \Pd = 14$ dBm, $M = 1$, $N = 50$, $B = 50$ bytes and $\beta = 0.1$.}
\label{fig:T1_tradeoff}
\end{center}
\end{minipage}\qquad
\begin{minipage}[b]{.47\textwidth}
\begin{center}
\psfrag{xlabel}[c][c][\scalevalueS]{number of pilot symbols per device, $L$}
\psfrag{ylabel}[c][c][\scalevalueS]{system outage probability, $\Pout$}
\psfrag{XXXXx1}[lc][lc][\scalevalueSS]{best case}
\includegraphics[width=\scalevalueFigWidthSmall\columnwidth,keepaspectratio]{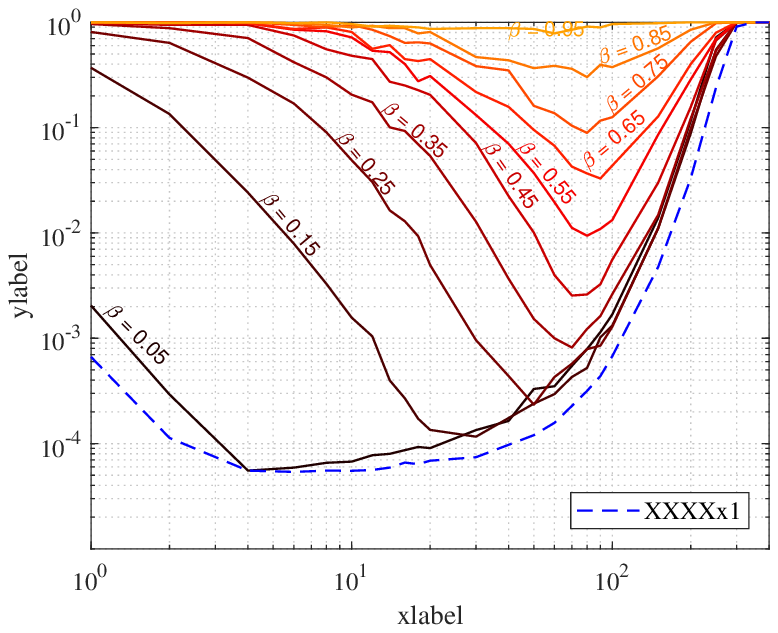}
\caption{System outage probability of the proposed \gls{andcoop} for different  $\beta$ values. We fixed $\Pa = \Pd = 14$ dBm, $M = 1$, $N = 50$, $B = 50$ bytes and $\rbo = 0.6$.}
\label{fig:Rbo_tradeoff}
\end{center}
\end{minipage}
\if\mycmd 1
\vspace{-35pt}
\else
\fi
\end{figure*}
\fi

\if\mycmd 2
\begin{figure}[t]
\begin{center}
\psfrag{xlabel}[c][c][\scalevalueS]{number of pilot symbols per device, $L$}
\psfrag{ylabel}[c][c][\scalevalueS]{system outage probability, $\Pout$}
\psfrag{XXXXx1}[lc][lc][\scalevalueSS]{best case}
\includegraphics[width=\scalevalueFigWidthSmall\columnwidth,keepaspectratio]{T1_fixed_tradeoff.eps}
\caption{System outage probability of the proposed \gls{andcoop} for different  $\rbo$ values. We fixed $\Pa = \Pd = 14$ dBm, $M = 1$, $N = 50$, $B = 50$ bytes and $\beta = 0.1$.}
\label{fig:T1_tradeoff}
\end{center}
\if\mycmd 1
\vspace{-35pt}
\else
\fi
\end{figure}

\begin{figure}[t]
\begin{center}
\psfrag{xlabel}[c][c][\scalevalueS]{number of pilot symbols per device, $L$}
\psfrag{ylabel}[c][c][\scalevalueS]{system outage probability, $\Pout$}
\psfrag{XXXXx1}[lc][lc][\scalevalueSS]{best case}
\includegraphics[width=\scalevalueFigWidthSmall\columnwidth,keepaspectratio]{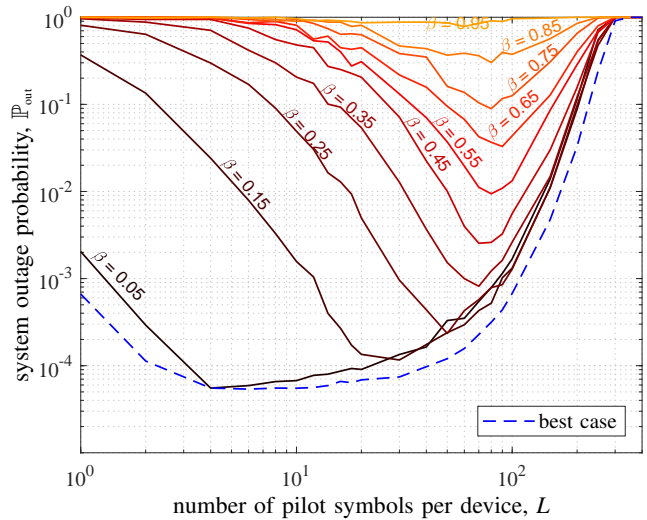}
\caption{System outage probability of the proposed \gls{andcoop} for different  $\beta$ values. We fixed $\Pa = Pd = 14$ dBm, $M = 1$, $N = 50$, $B = 50$ bytes and $\rbo = 0.6$.}
\label{fig:Rbo_tradeoff}
\end{center}
\if\mycmd 1
\vspace{-35pt}
\else
\fi
\end{figure}
\fi

\paragraph*{Fixing a small number of pilot symbols per device} To reduce the complexity of the exhaustive numerical optimization we restrict each of the three parameters  $L$, $\beta$ and $\rbo$, to a finite set of values. Then, for each triple, we simulate the system outage probability. The curves depicted against parameter $L$ in \figref{fig:T1_tradeoff}, show the course of system outage probability across different values of $\rbo$ when we fixed $\beta = 0.1$. Similarly, in \figref{fig:Rbo_tradeoff}, system outage probability is depicted for different $\beta$ values while fixing $\rbo = 0.6$. The best case system outage probability in both those figures, represents the minimum outage probability that is attainable at a given $L$ while optimizing against $\beta$ and $\rbo$. The following observations are given.
\begin{itemize}
\item 	By increasing $L$, channel estimation error decreases, which in turn lowers the impact of $\rbo$ on system  outage probability  for fixed $\beta$.
\item 	By  optimization across different values of $\beta$ and $\rbo$,  it becomes evident that system outage probability is within a small margin of the optimal value, across a wide range of $L$. In this example, choosing $2 \leq L \leq 30$, system outage probability remains roughly unchanged. 
\end{itemize}
From those observations, to simplify the optimization process we propose to fix $L = 10$ and  optimize only across $\beta$ and $\rbo$.

\if\mycmd 2
\begin{figure}[t]
\begin{center}
\psfrag{ylabel}[c][c][\scalevalueS]{probability of system outage, $\Pout$}
\psfrag{xlabel}[c][c][\scalevalueS]{transmit power (dBm)}
\psfrag{XXX1}[lc][lc][\scalevalueSS]{$\beta = \hat{\beta}$}
\psfrag{x2}[lc][lc][\scalevalueSS]{$\beta = 0$}
\psfrag{x3}[lc][lc][\scalevalueSS]{\gls{icsi}}
\psfrag{XXX4}[lc][lc][\scalevalueSS]{\gls{pcsi}}
\psfrag{XXX5}[lc][lc][\scalevalueSS]{$M = 1$}
\psfrag{x6}[lc][lc][\scalevalueSS]{$M = 3$}
\includegraphics[width=\scalevalueFigWidthSmall\columnwidth,keepaspectratio]{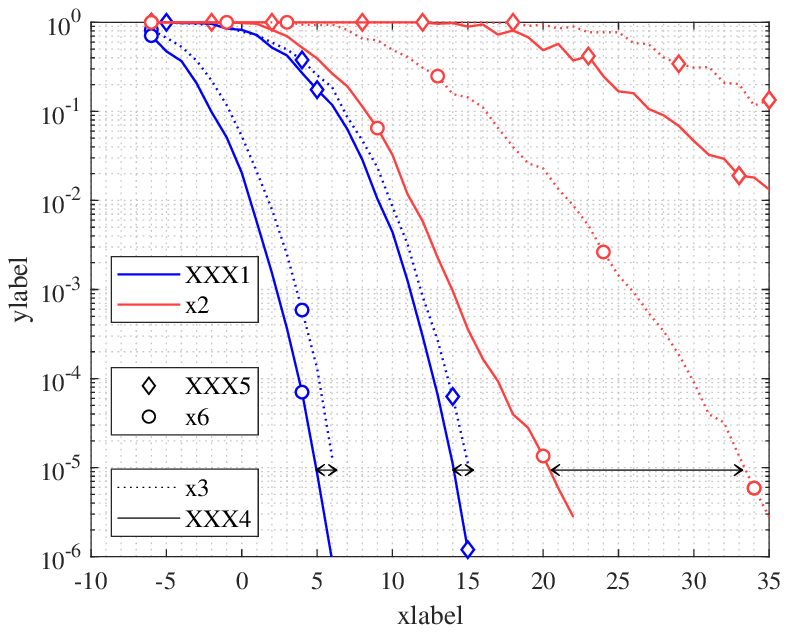}
\caption{Comparison of outage probability between \gls{icsi} and \gls{pcsi} for $N = 50$, $B = 50$ bytes per device and $L = 10$.}
\label{fig:comp_power_outage_reVSid}
\end{center}
\if\mycmd 1
\vspace{-35pt}
\else
\fi
\end{figure}
\fi

\if\mycmd 1
\begin{figure*}
\centering
\begin{minipage}[b]{.47\textwidth}
\begin{center}
\psfrag{ylabel}[c][c][\scalevalueS]{probability of system outage, $\Pout$}
\psfrag{xlabel}[c][c][\scalevalueS]{transmit power (dBm)}
\psfrag{XXX1}[lc][lc][\scalevalueSS]{$\beta = \hat{\beta}$}
\psfrag{x2}[lc][lc][\scalevalueSS]{$\beta = 0$}
\psfrag{x3}[lc][lc][\scalevalueSS]{\gls{icsi}}
\psfrag{XXX4}[lc][lc][\scalevalueSS]{\gls{pcsi}}
\psfrag{XXX5}[lc][lc][\scalevalueSS]{$M = 1$}
\psfrag{x6}[lc][lc][\scalevalueSS]{$M = 3$}
\includegraphics[width=\scalevalueFigWidthSmall\columnwidth,keepaspectratio]{comp_power_outage_reVSid.eps}
\caption{Comparison of outage probability between \gls{icsi} and \gls{pcsi} for $N = 50$, $B = 50$ bytes per device and $L = 10$.}
\label{fig:comp_power_outage_reVSid}
\end{center}
\end{minipage}\qquad
\begin{minipage}[b]{.47\textwidth}
\begin{center}
\psfrag{xlabel}[c][c][\scalevalueS]{system outage probability, $\Pout$}
\psfrag{ylabel}[c][c][\scalevalueS]{}
\psfrag{X1}[lc][lc][\scalevalueSS]{$\hat{\beta}$}
\psfrag{x2}[lc][lc][\scalevalueSS]{$\hat{\rbo}$}
\psfrag{XXX4}[lc][lc][\scalevalueSS]{$M = 1$}
\psfrag{x6}[lc][lc][\scalevalueSS]{$M = 3$}
\includegraphics[width=\scalevalueFigWidthSmall\columnwidth,keepaspectratio]{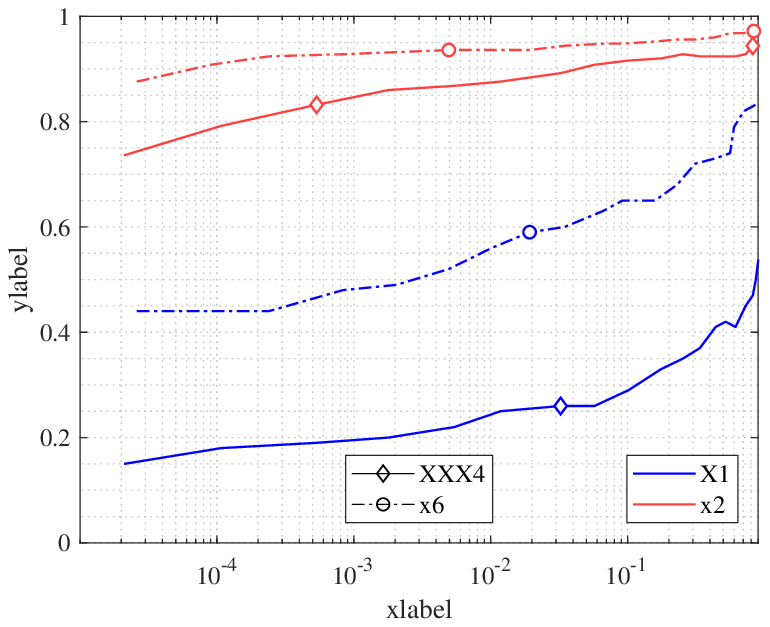}
\caption{The optimal values $\hat{\beta}$ and $\hat{\rbo}$ for \gls{icsi}, where $N = 50$, $B = 50$ bytes per device and $L = 10$.}
\label{fig:opt_beta_rbo_re}
\end{center}
\end{minipage}
\if\mycmd 1
\vspace{-35pt}
\else
\fi
\end{figure*}
\fi

\paragraph*{Impact of channel estimation error is marginal}  For fixed number of pilot symbols per device $L = 10$, we examine the performance degradation of the proposed \gls{andcoop} from \gls{icsi} with respect to the case of \gls{pcsi} under similar setup as in \figref{fig:comp_id_power}. As shown, for both cases of $M = 1$ and $M = 3$, in presence of \gls{icsi} the proposed scheme with $\beta = \hat{\beta}$ operates within a small 1-2 dB gap from the case of \gls{pcsi}. On the other hand, the gap for the case of single-hop rate-adaptive transmission can become very large, e.g., $~15$ dB for the case of $M = 3$. Such large \gls{snr} gap owes to high channel estimation error of the devices with poor channel quality during rate adaptation for  single-hop transmission. Our proposed \gls{andcoop} circumvents that by identifying those  devices and scheduling them over fixed-rate two-hop transmission.

\if\mycmd 2
\begin{figure}[h]
\begin{center}
\psfrag{xlabel}[c][c][\scalevalueS]{system outage probability, $\Pout$}
\psfrag{ylabel}[c][c][\scalevalueS]{}
\psfrag{X1}[lc][lc][\scalevalueSS]{$\hat{\beta}$}
\psfrag{x2}[lc][lc][\scalevalueSS]{$\hat{\rbo}$}
\psfrag{XXX4}[lc][lc][\scalevalueSS]{$M = 1$}
\psfrag{x6}[lc][lc][\scalevalueSS]{$M = 3$}
\includegraphics[width=\scalevalueFigWidthSmall\columnwidth,keepaspectratio]{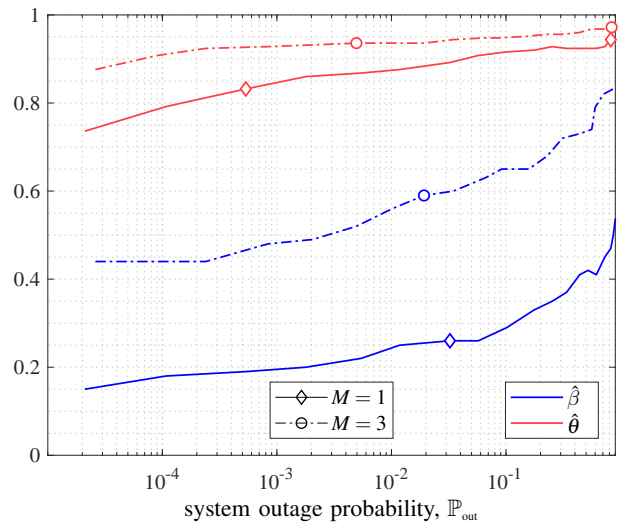}
\caption{The optimal values $\hat{\beta}$ and $\hat{\rbo}$ for \gls{icsi}, where $N = 50$, $B = 50$ bytes per device and $L = 10$.}
\label{fig:opt_beta_rbo_re}
\end{center}
\if\mycmd 1
\vspace{-35pt}
\else
\fi
\end{figure}
\fi

\paragraph*{Optimal value for $\beta$ and $\rbo$ for fixed $L$} In  \figref{fig:opt_beta_rbo_re}, the optimal  values $\hat{\beta}$ and $\hat{\rbo}$ are reported for fixed $L = 10$. Those values were used in \figref{fig:comp_power_outage_reVSid} for the proposed \gls{andcoop} in case of \gls{icsi}. It is evident that with a larger number of \gls{ap} antennas, on average a larger number of devices are scheduled with single-hop transmission (i.e., larger $\hat{\beta}$). Moreover, the optimal rate adjustment factor $\hat{\rbo}$ is smaller, when the number of \gls{ap} antennas is smaller. As could be expected, both  $\hat{\beta}$ and $\hat{\rbo}$ have a non-increasing trend with system outage probability decreasing. This parallels our design analogy for ultra-reliable communication, where only the devices with strong channel conditions should undergo rate-adaptive transmission (i.e., smaller $\hat{\beta}$), and for those, a more conservative transmission rate adjustment factor is necessary (i.e., smaller $\hat{\rbo}$).



\section{Conclusion}
\label{Sec:Conclusion}

We proposed a channel-aware  \gls{urllc} transmission protocol for  industrial wireless control where a controller communicates with several devices in the \gls{dl}. The proposed transmission protocol uses the knowledge of instantaneous channel conditions of \gls{ap}-device links to identify devices with \emph{strong} and \emph{weak} channel conditions. The strong devices are served with a single-hop rate-adaptive transmission  where the  rate is adapted to the instantaneous channel. With this approach, multi-user diversity gain is exploited and frequency resources are efficiently utilized. Meanwhile, the weak devices enjoy a two-hop cooperative communication in which transmission rate is fixed and all the nodes in the network cooperate in  relaying. We analyzed the system outage probability and the diversity order  under the proposed transmission protocol. Through numerical analysis we derived optimal time division between the two sets of \emph{strong} and \emph{weak} devices and showed that  such optimization can improve spectral efficiency by more than 0.5, 1 and 1.5 \gls{bpcu}, when the controller is equipped with 1, 2 and 3 \gls{ap} antennas, respectively. Moreover, we observed that the proposed \gls{andcoop} transmission protocol can effectively reduce the required transmit power for reliable industrial wireless control, improve the scalability with respect to network size, and marginalize the impact of channel estimation error on outage probability. The improvements are thanks to the instantaneous awareness to channel conditions that allows to exploit different sources of diversity gain in the network. As the continuation of this work in future, we will look into evaluating the performance of the proposed transmission protocol  considering  spatio-temporal correlation of shadowing caused by blockages. Further, we will study the use of dedicated full-duplex relay nodes in the proposed transmission protocol.

\input{app.tex}

\bibliographystyle{IEEEtran}
\bibliography{IEEEabrv,references_all}{}
\end{document}

%% file: include.tex
\usepackage{mathptmx}
\usepackage[mathscr]{eucal}

\usepackage{amsthm}
\usepackage{amsfonts}
\usepackage{amssymb}
\usepackage{amsmath}
\usepackage{mathtools}
\DeclareMathAlphabet{\mathcal}{OMS}{cmsy}{m}{n}

\DeclareSymbolFont{Letters}{OML}{cmm}{m}{it}
\DeclareMathSymbol{\alpha}{\mathalpha}{Letters}{11}
\DeclareMathSymbol{\beta}{\mathalpha}{Letters}{12}
\DeclareMathSymbol{\lambda}{\mathalpha}{Letters}{21}
\DeclareMathSymbol{\Lambda}{\mathalpha}{Letters}{3}
\DeclareMathSymbol{\pi}{\mathalpha}{Letters}{25}
\DeclareMathSymbol{\rho}{\mathalpha}{Letters}{26}
\DeclareMathSymbol{\sigma}{\mathalpha}{Letters}{27}

\makeatletter
\def\@IEEEsectpunct{:\ \,}
\def\paragraph{\@startsection{paragraph}{4}{\z@}{1.5ex plus 1.5ex minus 0.5ex}%
{0ex}{\normalfont\normalsize\bfseries}}
\makeatother

\def \scalevalueS {0.9}
\def \scalevalueSS {0.59}

\def \scalevalueFigWidth {0.8}

\usetikzlibrary{calc}
\usetikzlibrary{decorations.pathreplacing,decorations.markings,shapes.geometric}
\tikzset{naming/.style={align=center,font=\small}}
\tikzset{antenna/.style={insert path={-- coordinate (ant#1) ++(0,0.25) -- +(135:0.25) + (0,0) -- +(45:0.25)}}}
\tikzset{station/.style={naming,draw,shape=dart,shape border rotate=90, minimum width=10mm, minimum height=10mm,outer sep=0pt,inner sep=3pt}}
\tikzset{mobile/.style={naming,draw,shape=rectangle,minimum width=12mm,minimum height=6mm, outer sep=0pt,inner sep=3pt}}
\tikzset{radiation/.style={{decorate,decoration={expanding waves,angle=90,segment length=4pt}}}}

\tikzset{radiation/.style={{decorate,decoration={expanding waves,angle=90,segment length=4pt}}},
         BS/.pic={
        code={\tikzset{scale=5/10}
            \draw[semithick] (0,0) -- (1,4);
            \draw[semithick] (3,0) -- (2,4);
            \draw[semithick] (0,0) arc (180:0:1.5 and -0.5) node[above, midway]{#1};
            \node[inner sep=4pt] (circ) at (1.5,5.5) {};
            \draw[semithick] (1.5,5.5) circle(8pt);
            \draw[semithick] (1.5,5.5cm-8pt) -- (1.5,4);
            \draw[semithick] (1.5,4) ellipse (0.5 and 0.166);
            \draw[semithick,radiation,decoration={angle=50}] (1.5cm+8pt,5.5) -- +(0:2);
            \draw[semithick,radiation,decoration={angle=50}] (1.5cm-8pt,5.5) -- +(180:2);
  }}
}

\tikzset{radiation/.style={{decorate,decoration={expanding waves,angle=90,segment length=4pt}}},
         node/.pic={
        code={\tikzset{scale=5/10}
            \filldraw[semithick] (0,0) circle(10pt);
            \draw[semithick,radiation,decoration={angle=45}] (0,12pt) -- +(90:2);
  }}
}

\newcommand{\plos}{p_{\text{L}}}

\newcommand{\Pt}{P_t}
\newcommand{\Pa}{P_a}
\newcommand{\Pd}{P_d}

\newcommand{\Toh}{T_{1\text{h}}} 
\newcommand{\Tth}{T_{2\text{h}}} 
\newcommand{\Roh}[1]{R_{1\text{h},#1}} 
\newcommand{\Rth}[1]{R_{2\text{h},#1}} 
\newcommand{\aR}{{c}} 
\newcommand{\ahR}{{\hat{c}}} 

\newcommand{\ha}[2]{h^\tr{a}_{#1,#2}}
\newcommand{\hha}[2]{\hat{h}^\tr{a}_{#1,#2}}
\newcommand{\hd}[2]{h^\tr{d}_{#1,#2}}

\newcommand{\ga}[2]{g^\tr{a}_{#1,#2}}
\newcommand{\gha}[2]{\hat{g}^\tr{a}_{#1,#2}}
\newcommand{\gd}[2]{g^\tr{d}_{#1,#2}}

\newcommand{\esterror}{\epsilon}

\newcommand{\ra}[2]{\rho^\tr{a}_{#1,#2}}
\newcommand{\rd}[2]{\rho^\tr{d}_{#1,#2}}

\newcommand{\koh}{K_{1\text{h}}}
\newcommand{\kth}{K_{2\text{h}}}

\newcommand{\varnoise}{\sigma_{\tr{o}}}

\newcommand{\prob}[1]{\text{Pr}\left(#1\right)}

\newcommand{\indic}[1]{\mathbf{1}_{\scriptstyle #1}}

\newcommand{\Poh}{P_{1h}}
\newcommand{\Pth}{P_{2h}}
\newcommand{\Pm}[1]{p{\left(#1\right)}}

\newcommand{\rbo}{\theta}

\newcommand{\msa}{\mathscr{A}}

\newcommand{\msd}{\mathscr{D}}
\newcommand{\msc}{\mathscr{C}}
\newcommand{\mss}{\mathscr{S}}

\newcommand{\msr}{\mathscr{R}}

\newcommand{\msoh}{\mathscr{D}_{1\text{h}}}
\newcommand{\msth}{\mathscr{D}_{2\text{h}}}

\newcolumntype{C}{>{\(\displaystyle}c<{\)}@{}} 
\newcolumntype{L}{>{\(\displaystyle}l<{\)}@{}} 
\newcolumntype{R}{>{\(\displaystyle}r<{\)}@{}}

\newcommand{\tr}[1]{\textrm{#1}}

\newcommand{\set}[1]{\{#1\}}

\newcommand{\cdf}{\tr{F}}       

\newcommand{\card}[1]{|#1|}
\newcommand{\secref}[1]{Section~\ref{#1}}

\newcommand{\figref}[1]{Fig.~\ref{#1}}
\newcommand{\propref}[1]{Proposition~\ref{#1}}
\newcommand{\tabref}[1]{Table~\ref{#1}}

\newcommand{\appref}[1]{Appendix~\ref{#1}}

\newcommand{\Pout}{\mathbb{P}_\tr{\tiny{out}}}

\newtheorem{proposition}{Proposition}

%% file: Acronyms.tex
\newacronym{csi}{CSI}{channel state information}
\newacronym{cqi}{CQI}{channel quality indicator}
\newacronym{ack}{ACK}{acknowledgement}
\newacronym{arq}{ARQ}{automatic repeat request}
\newacronym{awgn}{AWGN}{additive white Gaussian noise}
\newacronym{cc}{CC}{chase combining}
\newacronym{dp}{DP}{dynamic programming}
\newacronym{fec}{FEC}{forward error correction}
\newacronym{harq}{HARQ}{hybrid automatic repeat request}
\newacronym{hspa}{HSPA}{high speed packet access}
\newacronym{iid}{i.i.d.}{independent and identically distributed}
\newacronym{ir}{IR}{incremental redundancy}
\newacronym{lte}{LTE}{long term evolution}
\newacronym{mdp}{MDP}{markov decision process}
\newacronym{mrc}{MRC}{maximal-ratio combining}
\newacronym{nack}{NAK}{negative acknowledgement}
\newacronym{pdf}{pdf}{probability density function}
\newacronym{wimax}{WiMax}{worldwide interoperability for microwave access}
\newacronym{3gpp}{3GPP}{3rd generation partnership project}
\newacronym{ofdm}{OFDM}{orthogonal frequency-division multiplexing}
\newacronym{ofdma}{OFDMA}{orthogonal frequency-division multiple access}
\newacronym{wlan}{WLAN}{wireless local area network}
\newacronym{mmse}{MMSE}{minimum mean-square-error}
\newacronym{gsm}{GSM}{global system for mobile communications}
\newacronym{edge}{EDGE}{enhanced data \gls{gsm} environment}
\newacronym{stbc}{STBC}{space-time block code}
\newacronym{amc}{AMC}{adaptive modulation and coding}
\newacronym{snr}{SNR}{signal to noise ratio}
\newacronym{sinr}{SINR}{signal to interference and noise ratio}
\newacronym{mi}{MI}{mutual information}
\newacronym{acmi}{ACMI}{accumulated mutual information}
\newacronym{nacmi}{NACMI}{normalized ACMI}
\newacronym{cdi}{CDI}{channel distribution information}
\newacronym{latr}{LATR}{long-term average transmission rate}
\newacronym{rtr}{RTR}{round transmission rate}
\newacronym{pomdp}{POMDP}{Partially Observable Markov Decision Process}
\newacronym{fd}{FD}{full-duplex}
\newacronym{hd}{HD}{half-duplex}
\newacronym{td}{TD}{Time Division}
\newacronym{tdma}{TDMA}{time-division multiple access}
\newacronym{mac}{MAC}{Media Access Control}
\newacronym{uwb}{UWB}{Ultra Wideband}
\newacronym{ieee}{IEEE}{institute of electrical and electronics engineers}
\newacronym{dB}{dB}{decibel}
\newacronym{cdf}{cdf}{cumulative density function}
\newacronym{ccdf}{ccdf}{complementary cumulative density function}
\newacronym{min}{Min.}{minimum}
\newacronym{med}{Med.}{median}
\newacronym{avg}{Avg.}{average}
\newacronym{ul}{UL}{up-link}
\newacronym{dl}{DL}{down-link}
\newacronym{app}{APP}{a-posteriori probability}
\newacronym{logmap}{LogMAP}{log maximum a-posteriori}
\newacronym{llr}{LLR}{log-likelihood ratio}
\newacronym{ue}{UE}{user equipment}
\newacronym{qos}{QoS}{quality of service}
\newacronym{5g}{5G}{fifth generation}
\newacronym{4g}{4G}{fourth generation}
\newacronym{tti}{TTI}{transmission time interval}
\newacronym{rrm}{RRM}{radio resource management}
\newacronym{mmib}{MMIB}{mean mutual information per bit}
\newacronym{dsi}{DSI}{decoder state information}
\newacronym{tb}{TB}{transport block}
\newacronym{tbs}{TBS}{transport block size}
\newacronym{cb}{CB}{code block}
\newacronym{cbg}{CBG}{code block group}
\newacronym{cbs}{CBS}{code block size}
\newacronym{prb}{PRB}{physical resource block}
\newacronym{rb}{RB}{resource block}
\newacronym{bler}{BLER}{block error rate}
\newacronym{blep}{BLEP}{block error probability}
\newacronym{crc}{CRC}{cyclic redundancy check}
\newacronym{tdd}{TDD}{time division duplex}
\newacronym{fdd}{FDD}{frequency division duplex}
\newacronym{embb}{eMBB}{enhanced mobile broadband}
\newacronym{mcc}{MCC}{mission critical communication}
\newacronym{mmc}{MMC}{massive machine communication}
\newacronym{mtc}{MTC}{machine type of communication}
\newacronym{mmtc}{mMTC}{massive machine type of communication}
\newacronym{umtc}{uMTC}{ultra-reliable \gls{mtc}}
\newacronym{urllc}{URLLC}{ultra-reliable low-latency communications}
\newacronym{rtt}{RTT}{round trip time}
\newacronym{rs}{RS}{reference symbols}
\newacronym{kpi}{KPI}{key performance indicator}
\newacronym{kpis}{KPIs}{key performance indicators}
\newacronym{tx}{Tx}{transmitter node}
\newacronym{rx}{Rx}{receiver node}
\newacronym{cran}{C-RAN}{centralized radio access network}
\newacronym{rru}{RRU}{remote radio unit}
\newacronym{bbu}{BBU}{baseband unit}
\newacronym{fhd}{FHD}{fronthaul delay}
\newacronym{cch}{CCH}{control channel}
\newacronym{saw}{SAW}{stop-and-wait}
\newacronym{qci}{QCI}{\gls{qos} class identifier}
\newacronym{gbr}{GBR}{guaranteed bit rate}
\newacronym{mbr}{MBR}{maximum bit rate}
\newacronym{ngbr}{non-GBR}{non-\gls{gbr}}
\newacronym{arp}{ARP}{allocation and retention priority}
\newacronym{effcr}{ECR}{effective coding rate}
\newacronym{cbs}{CBS}{code block size}
\newacronym{mcs}{MCS}{modulation and coding scheme}
\newacronym{eva}{EVA}{extended vehicular A}
\newacronym{epa}{EPA}{extended pedestrian A}
\newacronym{etu}{ETU}{extended typical urban}
\newacronym{re}{RE}{resource element}
\newacronym{reS}{REs}{resource elements}
\newacronym{nr}{NR}{new radio}
\newacronym{qpsk}{QPSK}{quadrature phase shift keying}
\newacronym{qam}{QAM}{quadrature amplitude modulation}
\newacronym{siso}{SISO}{single-input and single-output}	
\newacronym{miso}{MISO}{multiple-input single-output}
\newacronym{mimo}{MIMO}{multiple-input multiple-output}
\newacronym{bs}{BS}{base station}
\newacronym{phy}{PHY}{physical layer}
\newacronym{rlc}{RLC}{radio link control}
\newacronym{bcfsaw}{BCF-SAW}{BCF-SAW}
\newacronym{bcf}{BCF}{backwards composite feedback}
\newacronym{bac}{BAC}{binary asymmetric channel}
\newacronym{bsc}{BSC}{binary symmetric channel}
\newacronym{dtx}{DTX}{discontinued transmission}
\newacronym{bpsk}{BPSK}{binary phase shift keying}
\newacronym{bep}{BEP}{bit error probability}
\newacronym{ndi}{NDI}{new data indicator}
\newacronym{dci}{DCI}{downlink control information}
\newacronym{csit}{CSIT}{channel state information at the transmitter}
\newacronym{lt}{LT}{loudest talker}
\newacronym{ct}{CT}{cooperative transmission}
\newacronym{bps}{bps}{bits per second}
\newacronym{bpcu}{bpcu}{bits per channel use}
\newacronym{los}{LOS}{line-of-sight}
\newacronym{nlos}{NLOS}{non-line-of-sight}
\newacronym{regsaw}{Reg-SAW}{Regular SAW}
\newacronym{Lrep}{$L$-Rep-ACK}{Increased feedback repetition order}
\newacronym{Lack}{$L$-ACK-SAW}{$L$ required ACK per packet}
\newacronym{Asym}{Asym-SAW}{Asymmetric feedback detection for SAW}
\newacronym{bretx}{Blind-ReTx}{Blind retransmission}
\newacronym{df}{DF}{decode-and-forward}
\newacronym{af}{AF}{amplify-and-forward}
\newacronym{ap}{AP}{access point}
\newacronym{icn}{ICN}{industrial control network}
\newacronym{comp}{CoMP}{coordinated multi-point}
\newacronym{rhs}{RHS}{right hand side}
\newacronym{lhs}{LHS}{left hand side}
\newacronym{ap}{AP}{access point}
\newacronym{comp}{CoMP}{Coordinated Multipoint}
\newacronym{sumiso}{SU-MISO}{single-user multiple-input-single-output}
\newacronym{ibl}{IBL}{infinite block length}
\newacronym{fbl}{FBL}{finite block length}
\newacronym{icn}{ICN}{industrial control network}
\newacronym{lan}{LAN}{local area network}
\newacronym{wsn}{WSN}{wireless sensor network}
\newacronym{rt}{RT}{real-time}
\newacronym{tdm}{TDM}{time division multeplxing}
\newacronym{isi}{ISI}{inter-symbol interference}
\newacronym{nist}{NIST}{National Institute of Standards and Technology}
\newacronym{cbrs}{CBRS}{Citizens Broadband Radio Service}
\newacronym{los}{LOS}{line-of-sight}
\newacronym{nlos}{NLOS}{non-line-of-sight}
\newacronym{itu}{ITU}{International Telecommunications Union}
\newacronym{mmwave}{mmWave}{millimeter-wave}
\newacronym{nsr}{NSR}{noise-to-signal ratio}
\newacronym{das}{DAS}{distributed antenna system}
\newacronym{pdf}{PDF}{probability density function}		
\newacronym{pmf}{PMF}{probability mass function}
\newacronym{srs}{SRS}{sounding reference signal}
\newacronym{dmrs}{DMRS}{demodulation reference signal}
\newacronym{psd}{PSD}{power spectral density}
\newacronym{rf}{RF}{radio frequency}
\newacronym{id}{ID}{identifier}
\newacronym{df}{DF}{decode-and-forward}
\newacronym{iot}{IoT}{internet of things}
\newacronym{iiot}{IIoT}{industrial internet-of-things}
\newacronym{icsi}{I-CSI}{imperfect CSI}
\newacronym{pcsi}{P-CSI}{perfect CSI}
\newacronym{andcoop}{ANDCoop}{adaptive network-device cooperation}
\newacronym{dmt}{DMT}{diversity-multiplexing tradeoff}
\newacronym{scs}{SCS}{sub-carrier spacing}


%% file: table.tex
\begin{table*}[t]
\if\mycmd 1
\scriptsize
\else
\fi
\centering
\caption{Summary of notation.}
\if\mycmd 1
\vspace{-5pt}
\else
\fi
\label{table:notation}
\begin{tabularx}{\textwidth}{|l|X|} 
\hline
\textbf{Notation} 			& \textbf{Description} \\
\hline
$M; \; N$					& Total number of transmitting \glspl{ap}; total number of receiving devices. \\
\hline
$B$						& Payload size per device, in bytes. 					\\
\hline
$W; \; \lambda$ 			& Available bandwidth; the wavelength of radio signals. \\
\hline
$T;\; T_D; \; T_S; \; T_P$ 	& Cycle duration; downlink transmission time duration; symbol period of $1/W$; total pilot transmission period.	\\
\hline
$\Toh; \Tth$				& Duration of the single-hop rate-adaptive transmission phase; duration of the two-hop cooperative transmission phase. \\
\hline
$\beta; \; \alpha$			& Ratio of total downlink transmission time allotted to the single-hop rate-adaptive transmission phase; ratio of the two-hop cooperative transmission phase allotted to broadcasting. \\
\hline
$L$							& Number of uplink pilot symbols per device. 						\\
\hline
$\ha{i}{j}; \; \ga{i}{j}; \; \ra{i}{j}$ 		&  Channel (fading);  received \gls{snr}; average received \gls{snr}, between the $i$th \gls{ap} and $j$th device.  		\\
\hline
$\hd{k}{j}; \; \gd{k}{j}; \; \rd{k}{j}$ 		& Channel (fading); received \gls{snr}; average received \gls{snr}, between the $k$th and $j$th devices.  		\\
\hline
$\hat{h}; \;\esterror$ 		& 	Estimated channel fade; channel estimation error. \\
\hline
$\aR_{j}; \; \ahR_j$ 		& Achievable transmission rate for device $j$; estimated achievable transmission rate for device $j$. \\
\hline
$\msd; \; \msoh; \; \msth$	& Set of all devices; set of devices scheduled over single-hop rate-adaptive transmission phase; set of devices scheduled over two-hop cooperative transmission phase. \\
\hline
$\koh; \; \kth$ 			& $\card{\msoh}$; $\card{\msth}$. \\
\hline
$\varnoise; \; \sigma_e$	& \gls{awgn} noise power; channel estimation error power.\\
\hline
$\Pout(\cdot)$		& System outage probability. \\
\hline
$\Pt; \; \Pa; \; \Pd$		& Transmit power; transmit power of an access point; transmit power of a device. \\
\hline
$R; \; \eta$ 				& Transmission rate in \gls{bps}; spectral efficiency in \gls{bpcu}. \\
\hline
$\Roh{j}$ 					& Transmission rate for device $j \in \msoh$ in the single-hop rate-adaptive transmission phase. \\
\hline
$\Rth{b}; \; \Rth{r}$ 		& Transmission rate of broadcast; transmission rate of relaying, in the two-hop cooperative phase. \\
\hline
\end{tabularx}
\if\mycmd 1
\vspace{-25pt}
\else
\fi
\end{table*}

%% file: table2.tex
\if\mycmd 1
\def \scaletemp {0.45} 
\else
\def \scaletemp {1} 
\fi

\begin{table}[t]
\centering
\caption{Simulation parameters setup.}
\resizebox{\scaletemp\columnwidth}{!}{%
\begin{tabular}{ l  l}
\hline	\hline 	
Parameter description & Value \\
\hline	\hline 
Floor area & 100$\times$100 m$^2$  \\
Number of devices, $N$ & 50  \\
Number of \glspl{ap}, $M$  & 1, 2 or 3 \\
Data per device & 50 Bytes \\
Cycle duration, $T$ & 1 ms  \\
Bandwidth & 20 MHz \\
Carrier frequency & 3.5 GHz ($\lambda\approx$ 8.57 cm) \\
\gls{ap} transmit power, $\Pa$ & 23 dBm \\
device transmit power, $\Pd$ & 23 dBm \\
\gls{psd} of the \gls{awgn} & -174 dBm/Hz \\
Path loss exponent ($\nu \leq 10\lambda$) & 2 \\
\gls{los} path loss exponent ($\nu > 10\lambda$) & 3.26 \\
\gls{nlos} path loss exponent ($\nu > 10\lambda$) & 3.93 \\
Blockage model: probability parameter $a$ & 0.25 \\
Blockage model: cutoff parameter $b$ & 15 m \\
Shadowing power value: & \\
\hline
\quad \quad \gls{los}, from \gls{ap} antenna & 1.4 dB \\
\quad \quad \gls{nlos}, from \gls{ap} antenna & 4.6 dB \\
\quad \quad \gls{los}, from device antenna & 8.7 dB \\
\quad \quad \gls{nlos}, from device antenna & 15.2 dB \\
\hline \hline 
\end{tabular}
}
\label{Tab:Parameters}
\if\mycmd 1
\vspace{-25pt}
\else
\fi
\end{table}

%% file: app.tex
\if\mycmd 1
\vspace{-15pt} 
\fi
\begin{appendix} 


\subsection{Proof of Proposition 1}
\label{App:motiv}

The achievable rate $\aR_{j}$ for device $j$ is given as 
\if\mycmd 1
\fontsize{9}{11}\selectfont
\begin{align}
\aR_{j} = W \log \left( 1 + \sum_{i \in \msa} \ga{i}{j}  \right),
\end{align}
\normalsize
\else
\begin{align}
\aR_{j} = W \log \left( 1 + \sum_{i \in \msa} \ga{i}{j}  \right),
\end{align}
\fi
measured in \gls{bps}. Then, the random  time duration $\tau_j$, in seconds, required for successful transmission to device $j$  is equal to $\tau_j = \frac{B'}{\aR_{j}}$, where $B' = \frac{N}{K}B$ is the adjusted packet size per scheduled device for a given $K$. Note that  $\aR_{j}$'s are \gls{iid} random variables which result in \gls{iid} $\tau_j$'s. The transmission from source to $K$ arbitrarily chosen devices is then \emph{successful} if the sum of $\tau_j$'s for those $K$ devices is not larger than $T$. Without loss of generality, let's assume $\tau_1 \leq \tau_2 \leq \ldots \leq \tau_N$, meaning that $\aR_1 \geq \aR_2 \geq \ldots \geq \aR_N$, where $\aR_j$'s form \emph{order statistics} drawn from \gls{cdf} $\cdf_{\aR}$. 

For the sake of better reliability (i.e., maximum diversity gain), the scheduler  chooses the  $K$ devices with best channels to transmit to, where $K \in \{1, 2, \ldots, N\}$.  
Thus, the probability of transmission error is equivalent to the probability of \emph{time overflow} given as
\if\mycmd 1
\fontsize{9}{11}\selectfont
\begin{align}
\Pout  =  \prob{ \sum_{j = 1}^{K} \tau_j \geq T} = \prob{ \sum_{j = 1}^{K} \frac{1}{\aR_{j}} \geq  \frac{K T}{N B}}.
\label{Eq:Pout-k-best}
\end{align}
\normalsize
\else
\begin{align}\nonumber
\Pout & =  \prob{ \sum_{j = 1}^{K} \tau_j \geq T} \\
& = \prob{ \sum_{j = 1}^{K} \frac{1}{\aR_{j}} \geq  \frac{K T}{N B}}.
\label{Eq:Pout-k-best}
\end{align}
\fi

The diversity gain $d$ from \eqref{Eq:dive_order_def} can be derived for the case of \gls{iid} Rayleigh fading for all links and fixed average \gls{snr} of $\rho$ over all links, as follows. First, note that 
\if\mycmd 1
\fontsize{9}{11}\selectfont
\begin{align}
\prob{ \frac{1}{\aR_{j}} \geq  \frac{1}{R}}  \overset{(a)}{=} \sum_{k = N - j + 1}^{N} \binom{N}{k} \cdf_{\aR}(R)^{k} \left(1 -  \cdf_{\aR}(R) \right)^{N-k} 
 \overset{(b)}{=} \sum_{k = N - j + 1}^{N} \binom{N}{k} \Pm{M,R}^{k} \left(1 -  \Pm{M,R} \right)^{N-k},
\end{align}
\normalsize
\else
\begin{align}
\prob{ \frac{1}{\aR_{j}} \geq  \frac{1}{R}} & \overset{(a)}{=} \sum_{k = N - j + 1}^{N} \binom{N}{k} \cdf_{\aR}(R)^{k} \left(1 -  \cdf_{\aR}(R) \right)^{N-k} \\ \nonumber
& \overset{(b)}{=} \sum_{k = N - j + 1}^{N} \binom{N}{k} \Pm{M,R}^{k} \left(1 -  \Pm{M,R} \right)^{N-k},
\end{align}
\fi
where $(a)$ follows from the \gls{cdf} of the $N - j + 1$th order statistics \cite[Chapter~6]{Mittelhammer99} and $(b)$ follows from \eqref{Eq:pfail}. Using the approximations $\Pm{m,R} \approx (\frac{\omega}{\Pt})^m$ and $\left( 1 - \Pm{m,R} \right) \approx 1$ for $\frac{\omega}{\Pt} \rightarrow 0$ \cite{Laneman:2003}, the following holds for any bounded real value  $R$. 
\if\mycmd 1
\fontsize{9}{11}\selectfont
\begin{align}
d_j   = - \lim_{\Pt \rightarrow \infty} \frac{\log \prob{ \frac{1}{\aR_{j}} \geq  \frac{1}{R}}}{\log \Pt}  =  - \lim_{\Pt \rightarrow \infty} \frac{\log  \sum_{k = N - j + 1}^{N} \binom{N}{k} (\frac{\omega}{\Pt})^{k \cdot M} } {\log \Pt}  = M (N - j + 1)
\label{Eq:div-best-j}
\end{align}
\normalsize
\else
\begin{align}\nonumber
d_j  & = - \lim_{\Pt \rightarrow \infty} \frac{\log \prob{ \frac{1}{\aR_{j}} \geq  \frac{1}{R}}}{\log \Pt} \\ \nonumber
& =  - \lim_{\Pt \rightarrow \infty} \frac{\log  \sum_{k = N - j + 1}^{N} \binom{N}{k} (\frac{\omega}{\Pt})^{k \cdot M} } {\log \Pt} \\
& = M (N - j + 1)
\label{Eq:div-best-j}
\end{align}
\fi

For the case where $K = 1$ device with the best channels is transmitted to, the diversity gain follows from \eqref{Eq:div-best-j} as
\begin{align}
d = M N.
\end{align}
For  $K > 1$, the maximum diversity gain follows from the choice of $K$ devices with best channels where the probability of outage $\Pout$ from \eqref{Eq:Pout-k-best} is bounded as follows
\if\mycmd 1
\fontsize{9}{11}\selectfont
\begin{align}\nonumber
\prob{ \frac{1}{\aR_{K}} \geq  \frac{K T}{N B}}  \leq \prob{ \sum_{j = 1}^{K} \frac{1}{\aR_{j}} \geq \frac{K T}{N B}}  \leq  \prob{ \frac{1}{\aR_{K}} \geq  \frac{T}{N B}}.
\end{align}
\normalsize
\else
\begin{align}\nonumber
\prob{ \frac{1}{\aR_{K}} \geq  \frac{K T}{N B}}  \leq \prob{ \sum_{j = 1}^{K} \frac{1}{\aR_{j}} \geq \frac{K T}{N B}}  \leq  \prob{ \frac{1}{\aR_{K}} \geq  \frac{T}{N B}}.
\end{align}
\fi
Therefore, the diversity order of $\Pout$ is bounded on both sides by $M(N - K + 1)$ according to \eqref{Eq:div-best-j}, which concludes
\if\mycmd 1
\fontsize{9}{11}\selectfont
\begin{align}
d = M (N - K + 1).
\end{align}
\normalsize
\else
\begin{align}
d = M (N - K + 1).
\end{align}
\fi

\if\mycmd 1
\vspace{-15pt} 
\fi

\if\mycmd 1
\def \scaletemp {0.55} 
\else
\def \scaletemp {1} 
\fi

\subsection{Proof of Proposition 2}
\label{App:A}

Although a closed form of outage probability in \eqref{Eq:outage-1hop} is not available, the outage probability $\Pout  = \Poh(\Toh)$ is  bounded as follows.
\if\mycmd 1
\fontsize{9}{11}\selectfont
\begin{align}\label{Eq:bounds}
\resizebox{\scaletemp\columnwidth}{!}{$1 - \left(1 - \Pm{M,\frac{B}{\Toh}}\right)^N \leq \Pout \leq 1 - \left(1 - \Pm{M,\frac{N B}{\Toh}}\right)^N$}
\end{align}
\normalsize
\else
\begin{align}\label{Eq:bounds}
\resizebox{\scaletemp\columnwidth}{!}{$1 - \left(1 - \Pm{M,\frac{B}{\Toh}}\right)^N \leq \Pout \leq 1 - \left(1 - \Pm{M,\frac{N B}{\Toh}}\right)^N$}
\end{align}
\fi
The upper bound in \eqref{Eq:bounds} is realized by dividing the time $\Toh$ equally among the devices, while the lower bound is realized by allotting the time $\Toh$ to every device. Therefore, denoting the diversity order of the upper and the lower bounds in \eqref{Eq:bounds}, respectively, by $d_\text{\tiny{upper}}$ and $d_\text{\tiny{lower}}$, it can be concluded from \eqref{Eq:dive_order_def} that
\begin{align}\label{Eq:bounds_d}
d_\text{\tiny{upper}} \leq d \leq d_\text{\tiny{lower}}.
\end{align}
%
Using the approximation $\Pm{m,R} \approx (\frac{\omega}{\Pt})^m$ for $\frac{\omega}{\Pt} \rightarrow 0$ \cite{Laneman:2003}, the following holds for any bounded real value  $R$.
\if\mycmd 1
\fontsize{9}{11}\selectfont
\begin{align}\nonumber
 - \lim_{\Pt \rightarrow \infty} \frac{\log 1 - \left(1 - \Pm{M,R} \right)^N}{\log \Pt}  = - \lim_{\Pt \rightarrow \infty} \frac{\log N (\frac{\omega}{\Pt})^M}{\log \Pt} = M
\end{align}
\normalsize
\else
\begin{align}\nonumber
 - \lim_{\Pt \rightarrow \infty} \frac{\log 1 - \left(1 - \Pm{M,R} \right)^N}{\log \Pt}  = - \lim_{\Pt \rightarrow \infty} \frac{\log N (\frac{\omega}{\Pt})^M}{\log \Pt} = M
\end{align}
\fi
Therefore, we have  $d_\text{\tiny{upper}} = d_\text{\tiny{lower}} = M$, which according to \eqref{Eq:bounds_d} yields \eqref{Eq:div_order_1hop}.
\subsection{Proof of Proposition 3}
\label{Proof_prop_2h}

We start from  the probability of outage   $\Pout = \Pth(\msd)$ in \eqref{Eq:Ptwohop-4}. First note that for $\Pt \rightarrow \infty$, where $n>0$ we have
\if\mycmd 1
\fontsize{9}{11}\selectfont
\begin{align}\label{Eq:qr-prop}
\frac{\Pm{M+n,\Rth{r}} }{ \Pm{M,\Rth{b}} } = \frac{ \left( \frac{\omega_{2\text{h},r}}{\Pt} \right)^{M+n}}{\left( \frac{\omega_{2\text{h},b}}{\Pt} \right)^M },
\end{align}
\normalsize
\else
\begin{align}\label{Eq:qr-prop}
\frac{\Pm{M+n,\Rth{r}} }{ \Pm{M,\Rth{b}} } = \frac{ \left( \frac{\omega_{2\text{h},r}}{\Pt} \right)^{M+n}}{\left( \frac{\omega_{2\text{h},b}}{\Pt} \right)^M },
\end{align}
\fi
where $\omega_{2\text{h},b} = W \cdot \sigma_0 \cdot (2^{\Rth{b}/W} - 1)$, and $\omega_{2\text{h},r} = W \cdot \sigma_0 \cdot (2^{\Rth{r}/W} - 1)$, and
\if\mycmd 1
\fontsize{9}{11}\selectfont
\begin{align}
\Rth{b} & = \frac{NB}{\alpha  T} = \frac{W r \log \Pt}{\alpha}, \\
\Rth{r} & = \frac{NB} { (1-\alpha)  T} = \frac{W r \log \Pt}{1- \alpha}.
\end{align}
\normalsize
\else
\begin{align}
\Rth{b} & = \frac{NB}{\alpha  T} = \frac{W r \log \Pt}{\alpha}, \\
\Rth{r} & = \frac{NB} { (1-\alpha)  T} = \frac{W r \log \Pt}{1- \alpha}.
\end{align}
\fi
For the case where $n = 0$ and $\alpha \geq 0.5$, we have $q_{r}^{(M+n)} = 1$. Otherwise, $q_{r}^{(M+n)}$ is derived using \eqref{Eq:qr-prop}. Therefore, in the limit of $\Pt \rightarrow \infty$, we use the following approximation, 
\if\mycmd 1
\fontsize{9}{11}\selectfont
\begin{align}\label{Eq:a-qr-approx}
1 - \left( 1 - {q_{r}^{(M+n)}} \right)^{(N-n)} \approx 
\left\{\begin{matrix}
1 & n = 0 \; \& \; \alpha \geq \frac{1}{2} \\ 
(N-n) q_{r}^{(M+n)} & \tr{otherwise}
\end{matrix}\right.
\end{align}
\normalsize
\else
\begin{align}\label{Eq:a-qr-approx}
1 - \left( 1 - {q_{r}^{(M+n)}} \right)^{(N-n)} \approx 
\left\{\begin{matrix}
1 & n = 0 \; \& \; \alpha \geq \frac{1}{2} \\ 
(N-n) q_{r}^{(M+n)} & \tr{otherwise}
\end{matrix}\right.
\end{align}
\fi

We further  use the approximation $\left(1-q_{b}^M \right)^{n} \approx 1$ in high \gls{snr}. Thus,    we obtain
\if\mycmd 1
\fontsize{9}{11}\selectfont
\begin{align}\label{Eq:two-hop-approx-outage-div}
\Pout & \approx \sum_{n = 0}^{N-1}\left( q_{b}^{M}\right)^{(N-n)}  \binom{N}{n} \left( 1 - \left( 1 - {q_{r}^{(M+n)}} \right)^{(N-n)} \right).
\end{align}
\normalsize
\else
\begin{align}\label{Eq:two-hop-approx-outage-div}
\Pout & \approx \sum_{n = 0}^{N-1}\left( q_{b}^{M}\right)^{(N-n)}  \binom{N}{n} \left( 1 - \left( 1 - {q_{r}^{(M+n)}} \right)^{(N-n)} \right).
\end{align}
\fi
Using the approximation in \eqref{Eq:a-qr-approx} it easily follows that for $M > 1$, the slowest term in  \eqref{Eq:two-hop-approx-outage-div} approaching zero as $\Pt \rightarrow \infty$  is the term of $n = N - 1$. Therefore, we use this term to approximate $\Pout$, where  substituting  in \eqref{Eq:dive_order_def} yields \eqref{Eq:DMT-2hop}. 

\end{appendix}